\documentclass[11pt,letterpaper]{article}
\usepackage{amsmath, amsthm, amssymb}
\usepackage{graphicx}
\usepackage[latin1]{inputenc}

\newcommand{\dist}[3]{d(#1, #2)}

\newcommand{\bl}{\textrm{blocked}}
\newcommand{\un}{\textrm{unassigned}}

\newcommand{\remove}[1]{}
\newcommand{\commento}[1]{}

\newtheorem{definition}{Definition}

\newtheorem{remark}{Remark}

\newtheorem{lemma}{Lemma}
\newtheorem{theorem}{Theorem}

\begin{document}

\title{On the Complexity of Searching in Trees: Average-case Minimization}
\date{}

\author{Tobias Jacobs\\ Albert-Ludwigs-University Freiburg, Germany\\
{\tt jacobs@informatik.uni-freiburg.de} \and  Ferdinando Cicalese\\  University of Salerno, Italy\\
{\tt cicalese@dia.unisa.it} \and Eduardo Laber \\PUC-Rio, Brazil\\ {\tt laber@inf.puc-rio.br} \and Marco Molinaro\\ 
Carnegie Mellon, USA\\ {\tt molinaro@cmu.edu}}

\maketitle
\thispagestyle{empty}

\begin{abstract}
We study the following {\em tree search problem}: in a given  tree $T = (V, E)$ a node has been marked and we want to identify it.
In order to locate the marked node, we can use  edge queries. An edge query $e$ asks  
 in which of the two connected components of $T\setminus e$ the marked node lies.  
The worst-case scenario where one is  interested in minimizing the maximum number of queries is well understood, 
and linear time  algorithms are known for finding an optimal search strategy [Onak {\em et al.\ }FOCS'06, Mozes {\em et al.\ }SODA'08] . 
Here we study the more involved average-case analysis: 
A function  $w : V \rightarrow \mathbb{Z}^+$ is given which defines the likelihood for a node to be the one marked,
and we want the strategy that minimizes the expected number of queries. Prior to this paper, 
very little was known about this natural question and  the complexity of the problem had remained so far an open question.

We close this question  and  prove that the above {\em tree search problem} is ${\cal NP}$-complete even  for the class of trees
with  diameter at most 4.  This results in a complete characterization of the 
complexity of the problem with respect to the diameter size. In fact, for diameter not larger than $3$ 
the problem can be shown to be polynomially solvable using a 
 dynamic programming approach.

In addition we prove 
that  the problem is ${\cal NP}$-complete even for the class of trees of maximum degree at most $16$. 
To the best of our knowledge, the only known result in this direction is that 
 the tree search problem is solvable  in $O(|V| \log |V|)$ time for trees with degree at most 2 (paths).

We match  the above complexity results with a tight algorithmic analysis. We first show that a  natural greedy algorithm 
attains a $2$-approximation. Furthermore, for the bounded degree instances, 
we show that any optimal  strategy (i.e., one that minimizes the expected number
of queries) performs at most $ O( \Delta(T) ( \log |V| + \log w(T))) $ queries in the worst case, where
$w(T)$ is the  sum of the likelihoods  of the nodes of $T$ and $\Delta(T)$ is the maximum degree of $T$.
We combine this result with a non-trivial exponential time algorithm to provide
an FPTAS for trees with bounded degree. 
\end{abstract}

\pagebreak

\setcounter{page}{1}

\section{Introduction}
Searching is one of the  fundamental problems in Computer Science and 
Discrete Mathematics. In his classical book \cite{Knuth73}, D.\ Knuth
discusses many variants of the searching problem, most of them  dealing with  totally ordered sets.
There has been some effort to extend the available techniques
for searching and for other fundamental problems (e.g. sorting and selection)  
to handle more complex structures such as
partially ordered sets \cite{LinSak85,FaigleEtAl86,conf/focs/OnakP06,Mozes_SODAO8,conf/soda/DaskalakisKMRV09}. 	
Here,
we focus on searching in structures that lay between totally ordered sets and the most general posets.
We wish  to efficiently locate a particular node in a tree. 

More formally, as input we are given a  tree $T = (V, E)$ which has a `hidden' \emph{marked} node
and a function $w : V \rightarrow \mathbb{Z}^+$ that gives the likelihood of a node being the one marked.
In order to discover which node of $T$ is marked, we can perform \emph{edge queries}: after querying the edge $e \in E$ we receive an answer stating in which of the two  connected components of $T \setminus e$ the marked node lies. 
To simplify our notation let us assume that our input tree $T$ is  rooted at a node $r$ so that
we can specify a query to an edge $e=uv$, with $u$ being the parent of $v$, by referring to  $v$.

A search strategy is a procedure that decides the next query to be posed based on the outcome of the previous queries. 
Every search strategy for a tree $T=(V,E)$ (or for a forest) can be represented by a binary search (decision) tree $D$ such that a path
from the root of $D$ to a leaf $\ell$ indicates which queries should be made at each step to discover that $\ell$ is the marked node. More precisely, 
a search  tree for  $T$ is a triple $D = (N, E', A)$, where $N$ and $E'$ are the nodes and edges of a binary tree and the assignment $A :  N \rightarrow V$ satisfies the following properties: 
(a) for every node $v$ of $V$ there is exactly one leaf $\ell$ in $D$ such that $A(\ell) = v$;
(b)[search property] if $v$ is in the right (left) subtree of $u$ in $D$ then $A(v)$ is (not) in  the subtree of $T$ rooted
at $A(u)$. For an example we refer to Figure \ref{fig:prob-defi}.

Given  a search tree $D$ for $T$,  let $\dist{u}{v}{D}$ be the length (in number of edges) of the path from $u$ to $v$ in $D$. Then the cost of $D$, or alternatively the expected number of queries of $D$ is given by
	\begin{eqnarray*}
		cost(D) = \sum_{v \in leaves(D)} \dist{root(D)}{v}{D} w(A(v)) \;.
	\end{eqnarray*}
	
\vspace{-0.2cm}	
Therefore, our problem can be stated as follows: given a rooted tree $T=(V,E)$ with $|V|=n$  and a function $w: V \rightarrow \mathbb{Z}^+$, the goal is to compute a minimum cost search tree for $T$. This 
	 is a natural generalization of the problem of searching an element
in a sorted list with non-uniform access probabilities.


\remove{We proceed with these operations until the marked node is found. 
The expected cost  of  a  search strategy $\mathcal{S}$ is defined as  $\sum_{v \in V} s_v w(v)$,
where $s_v$ is the number of queries needed by $\mathcal{S}$ to find the marked node when $v$ is the marked node.
Therefore, our problem consists of designing a search  strategy (decision tree)  with minimum expected cost.
}

\remove{Explain the problem in POSET in  terms of edges
The problem of minimizing the worst case cost of searching for an element in a POSET
is known to be NP-COMPLETE.}

\medskip

\noindent
{\bf The State of the Art.} The variant of the problem in which the goal is to minimize the number of edge queries in the worst case, 
rather than minimizing the expected number of queries, has been studied 
in several recent papers \cite{BenFarNew99,conf/focs/OnakP06,Mozes_SODAO8}. 
It turns out that an 
optimal (worst-case) strategy can be found in linear time \cite{Mozes_SODAO8}. 
This is in great contrast with the state of the art (prior to this paper) about the 
average-case minimization we consider here. The known results amount  to
the $O( \log n)$-approximation obtained by  
Kosaraju {\em et al.\ }\cite{KosPrzBor99}, and Adler and Heeringa \cite{conf/approx/AdlerH08} 
for the much more general binary identification problem, and the constant factor  
approximation algorithm  that two of the authors  gave in  \cite{conf/icalp/LaberM08}.
However, the complexity of the average-case minimization of the tree search problem has so far remained unknown.
\remove{
In \cite{conf/icalp/LaberM08}, it is presented a constant factor approximation algorithm  
for searching in trees that relies on  heavy-path decompositions and entropy arguments.
These few and far from complete results are in great contrast with the state of the art in the case of the worst case analysis of the 
tree search problem, which is know to be optimally solvable in linear time. Notwithstanding its  
fundamental character, for the more involved average case analysis  of the question investigate here, dramatically state of the art is in great contrasts with the Apart from the above results, the complexity of searching in trees was unknown prior to our work.  
}
\medskip

\noindent {\bf Our Results.} 
We 
significantly narrow 
the gap of knowledge in the complexity landscape of the tree search problem under two different points of view. 
We 
prove that this  problem is ${\cal NP}$-Complete
even for the class of trees with diameter at most 4.
This  results in a complete characterization of the problem's complexity with respect to the parametrization in terms of the diameter. 
In fact, the problem can be shown to be polynomially solvable for the class of trees of diameter at most $3.$ 
We also show that the tree search problem under average minimization is ${\cal NP}$-Complete for trees of degree at most $16$
\enlargethispage{0.05in}
(note that in any infinite class of trees either the diameter or the degree is non-constant).
This substantially improves upon the state of the art,   the only known result in this direction being an 
$O(n \log n)$ time  solution \cite{HuTuc71,GarWac77} for the class of trees with maximum degree  2.
The hardness  results are obtained by  fairly involved reductions from the Exact 3-Set Cover (X3C) with multiplicity 3 \cite{Gar79a}.

In addition to the complexity results, we also significantly improve the previous known 
results from the algorithmic perspective. We first show that we can attain $2$-approximation by a
simple greedy approach that always seeks to divide the remaining tree as evenly as possible.
For bounded-degree  trees,  we match the new hardness results with 
 an  FPTAS. In order to obtain the  FPTAS,
 we first devise a non-trivial Dynamic Programming based 
algorithm that, roughly speaking, computes the best possible
search tree, among the search trees with height at most $H$, in
 $O(n^2 2^H)$ time. Then, we show that every tree $T$ admits 
a minimum cost search tree  whose height  is  
$O( \Delta \cdot ( \log n + \log w(T)))$, where $\Delta$ is the maximum degree of $T$ and $w(T)$ is the 
total weight of the nodes in $T$.  This bound is of independent interest because the height of any  search tree for a complete tree  of degree $\Delta$   is $\Omega(\frac{\Delta}{\log \Delta} \log n).$ Furthermore, 
it  allows us to execute  the DP algorithm with $H=c \cdot \Delta \cdot (\log n + \log w(T))$, for a suitable constant $c$, obtaining
a pseudo-polynomial time algorithm for trees with bounded degree. 
By scaling the weights $w$ in a fairly standard way we   obtain the FPTAS.


The worst-case scenario has also been studied for the case where a question is posed to some node $u$ 
and the answer is either that $u$ is the marked node or in which 
connected component of the forest $T \setminus \{u\}$ the marked node lies \cite{Schaffer89,conf/focs/OnakP06}.
We remark that it is possible to adapt our  techniques to prove that for the average-case minimization,  
this ``node query''-variant of the tree search problem 
is also ${\cal NP}$-Complete; furthermore, we can provide for it  a (degree independent) FPTAS . Due to the space constraints we have 
to defer these results to the full version of the paper.

\medskip

\noindent {\bf Other Related Work.}
Besides the above mentioned papers, the worst-case version  of searching in trees
had already been studied and solved  under a different name,  one decade ago, 		
as pointed out by Dereniowski \cite{journals/dam/Dereniowski08}.
That is because the problem of searching a node in a tree
is equivalent to the problem of ranking the edges 
of a tree \cite{IyeRatVij91,laTGreSch95,Lam98optimaledge}.

The problem studied here  can also be seen
 as a particular case of the binary identification problem (BIP) \cite{Garey:1972:OBI}.
Suppose we are given a set of elements $ U=\{u_1,\ldots,u_n\}$,  a 
 set of tests $\{t_1,\ldots,t_m\}$, with $t_i \subseteq U$,
  a `hidden' marked element and
 a likelihood function $w: U \mapsto \mathbb{R}^+$. 
 A test $ t$ allows to determine whether the marked element is in the set $t$ or in $U \setminus t$.
 The BIP  consists of defining a strategy (decision tree) that minimizes
 the (expected) number of tests to find the marked element.
Both the average-case  and the worst-case minimization are
${\cal NP}$-Complete \cite{HR76}, and none of them  admits an $o(\log n)$-approximation unless ${\cal P}={\cal NP}$ \cite{journals/dam/LaberN04,DTentity}.
For both versions, simple greedy algorithms attain $O(\log n)$-approximation \cite{KosPrzBor99,Arkin:1998:DTG,conf/approx/AdlerH08}.
When we impose some structure  in the set of tests  we have interesting particular cases.
If the set of tests consists of all the subsets of $U$ (i.e.,  $2^{U}$), then the strategy that minimizes the average cost 
is a Huffman tree.
Let $G$ be a DAG with vertex set $U$.
If the set of tests is $\{t_1,\ldots,t_n\}$, where $t_i=\{u_j | u_i \leadsto u_j \mbox{ in } G\}$, then
we have the problem of searching in a poset \cite{journals/tit/LipmanA95,KosPrzBor99,searchRandom}. 
When $G$ is a directed path we have the alphabetic coding problem
\cite{HuTuc71}. The problem we study here corresponds to the particular case where $G$ is a directed tree.

\smallskip

\noindent {\bf Applications.} The problem of searching in posets (and in particular in trees) has practical applications in file system synchronization and software testing 
according to \cite{BenFarNew99,Mozes_SODAO8}. 


Strategies for searching in trees have also potential application to asymmetric communication protocols \cite{AdlerEtAl06a,AdlMag01,saberi,LabHol02,watkinson}.
In this scenario, a client has to send a binary string $x \in \{0,1\}^t$ to the server.
 $x$ is  drawn from a probability distribution ${\cal D}$  only
available to the server. The asymmetry comes from the client having much larger bandwidth for downloading than for uploading. In order to benefit
from this discrepancy, both parties agree on a protocol to exchange bits until the
server learns the string $x$, trying to 
minimize the number of bits sent 
by the client (though other factors, e.g.,  the number of rounds should also be taken into account).
In one of the first protocols \cite{AdlMag01,LabHol02},
at each round the server sends 
a binary string $y$ and the client
replies with a 0 or 1 depending on  whether $y$ is a prefix of $x$ or not.
Based on the client's answer, the server updates his knowledge about $x$ and sends another string if he has
not learned $x$ yet. 
This protocol corresponds to a strategy
for searching a marked leaf in a complete binary tree of height $t$, 
where only the leaves  have non-zero probability. 
In fact, the  binary strings in $\{0,1\}^t$ can be represented by a complete binary tree of height $t$ 
where every edge that connects a node
to its left (right) child  is labeled with 0 (1). This gives a 1-1 correspondence between binary strings of length at most $t$
 and  edges of the tree, and
 the message $y$ sent by the
server naturally corresponds to an edge query.

\section{Hardness}
\label{sec:hardness}
In this section we shall prove that  
the tree search problem defined above 
is ${\cal NP}$-Complete.
We shall use a reduction from the Exact 3-Set Cover problem with multiplicity bounded by 3, i.e., each element of the ground set can appear in 
at most 3 sets. 

An instance of the $3$-bounded Exact 3-Set Cover problem (X3C)  is defined by: (a)  a set $U = \{u_1, \dots, u_n\},$ with 
$n=3k$ for some $k \geq 1;$ 
(b) a family  ${\cal X} = \{X_1, \dots, X_m\}$ of subsets of $U,$ such that $|X_i| = 3$ for each $i=1, \dots m$ and for each $j = 1, \dots n,$
we have that 
$u_j$ appears in at most $3$ sets of ${\cal X}.$ 
Given an instance ${\mathbb I}=(U, {\cal X})$ the X3C problem is to decide whether ${\cal X}$ contains a partition of $U,$ i.e., whether 
there exists a family  ${\cal C} \subseteq {\cal X}$ such that $|{\cal C}| = k$ and $\bigcup_{X \in {\cal C}} X = U.$
This problem is well known to be ${\cal NP}$-Complete \cite{Gar79a}.

For our reduction it will be crucial to define an order among the sets of the family  $\cal X.$ 
Any   total  order $<$ on $U,$ say $u_1 < u_2 < \dots < u_n,$ can be extended   to a total 
order $\prec$ on ${\cal X} \cup  U$ by stipulating that:
(a) for any $X = \{x_1, x_2, x_3\}, Y = \{y_1, y_2, y_3\} \in {\cal X}$ (with 
$x_1 < x_2 < x_3$ and $y_1 < y_2 < y_3,$) the relation $X \prec Y$ holds if and only if 
the sequence $x_3\, x_2 \, x_1$  is lexicographically smaller than  $y_3\, y_2 \, y_1;$ (b)
for every  $j=1, \dots, n,$ the relation $u_j \prec X$ holds if and only if
the sequence $ u_j \, u_1 \, u_1 $ is lexicographically smaller than $x_3\, x_2 \, x_1.$ 

\smallskip

Assume an order  $<$ on $U$ has been fixed and   $\prec$ is its extension to  $U \cup {\cal X},$ as defined above.    
We denote by  $\Pi = (\pi_1, \dots, \pi_{n+m})$ the sequence of elements of $U \cup {\cal X}$ sorted 
in increasing order according to $\prec.$
From now on, w.l.o.g., we assume  that according to $<$ and $\prec$ ,  it holds that
 $u_1 < \dots < u_n$ and    $X_1 \prec  \dots \prec X_m.$ 
For each $i=1, \dots, m,$ we shall denote the elements of $X_i$ by $u_{i \, 1}, u_{i \, 2}, u_{i\, 3} $ so that 
$u_{i\, 1} < u_{i\, 2} < u_{i\, 3}.$

\smallskip

\noindent
{\bf Example 1.}
 Let  $U = \{a, b, c, d, e, f\},$ and ${\cal X} = \{\{a,b,c\}, \{b,c,d\},\{d,e,f\}, \{b, e, f\} \}.$ 
Then, fixing the standard alphabetical order among the elements of
$U,$ we have that the sets of ${\cal X}$ are ordered  as follows: $X_1 = \{a,b,c\}, X_2 = \{b,c,d\}, X_3 = \{b, e, f\}, X_4= \{d,e,f\}.$
Then, we have  ${\Pi} = (\pi_1, \dots, \pi_{10}) = (a, b, c, X_1, d, X_2, e, f, X_3, X_4).$ 
\smallskip



Because of  the orders we fixed and  the 
fact that each element of $U$ appears in at most 3 sets of ${\cal X},$ it follows that 
that we cannot have more  than three sets of $ {\cal X}$ 
appearing consecutively in $\Pi$. This will be important to prove the hardness for  bounded degree instances.


We shall first show a polynomial time reduction that maps
any instance ${\mathbb I} = (U, {\cal X})$ of 3-bounded X3C to an instance $\mathbb{I}' = (T, w)$ of the tree search problem, such that 
$T$ has diameter 4 but unbounded degree. We will then modify such reduction and show  hardness for the bounded case too. 

 \smallskip

\noindent
{\bf The structure of the tree $T$}.  The root of $T$ is denoted by $r.$ 
For each $i = 1, \dots, m$ the set $X_i \in {\cal X}$ is mapped to a tree $T_i$ of height 1,  with root $r_i$ and leaves  $t_i, s_{i\,1}, s_{i\,2}, s_{i\,3}.$
In particular, for $j=1,2,3,$ we  say that $s_{i \, j}$ is associated with the element $u_{i\, j}.$
We make each $r_i$ a child of $r.$
For $i=1, \dots, m,$ we also create  four leaves $a_{i1},a_{i2},a_{i3}, a_{i4}$ and make them children of  the root $r.$ 
We also define $\tilde{X_i} = \{t_i, s_{i \, 1}, s_{i \, 2}, s_{i \, 3}, a_{i1},\dots, a_{i4}\}$ to be the set of leaves of $T$ associated with 
$X_i.$ 
For the example given above, the corresponding tree is given in Figure \ref{fig:tree}.

 \smallskip

\noindent
{\bf The weights of the nodes of $T$}. Only the leaves of $T$ will have non-zero weight, i.e., we set $w(r) = w(r_1) = \cdots = w(r_m) = 0.$ 
While defining the weight of the leaves of $T$  it will be useful to assign weight also to each  $u \in U.$ In particular, our weight assignment will be 
such that each leaf in $T$ which is associated with an element $u$ will be assigned the same weight we assign to $u.$ Also, when we fix the 
weight of $u$ we shall understand that we are fixing the weight of all leaves in $T$ associated with $u.$
 We extend the function $w()$ to sets, 
 so the weight of a set is the total weight of its elements. Also we define the weight of a tree as the total weight of  its nodes.

The  weights will be set in order to force any optimal search tree for $(T,w)$ to have a well-defined structure. The following 
notions of Configuration and Realization will be useful to describe such a structure of an optimal search tree. 
In describing the search tree we shall use $q_{\nu}$ to denote the node in the search tree under consideration that represents the 
question about the node $\nu$ of the input tree $T.$
Moreover, we shall in general only be concerned with the part of the search tree meant to identify the nodes of $T$ of non-zero weight.
It should be clear that  the search tree can be easily completed by appending the remaining queries at the bottom.

\begin{definition} \label{defi:seq}
Given leaves
$\ell_1, \dots, \ell_h$  of $T,$ a {\em sequential search tree for $\ell_1, \dots, \ell_h$} is a search tree of height $h$ whose 
left path is $q_{\ell_{1}}, \dots, q_{\ell_{h}}$. 
This is  the strategy that asks about one leaf after another until they have all been considered. 
See Figure \ref{fig:Config-LinearStrat} (a) for an example.
\end{definition}

\noindent
{\bf Configurations, and Realizations of $\Pi$.}
For each $i=1,\dots,m,$ 
let $D^{A}_i$ be the search tree with root $q_{r_i}$, with right subtree being the sequential search tree for 
$t_i, s_{i\,3}, s_{i\,2}, s_{i\,1},$ and left subtree  being a  sequential search tree for (some permutation of)
$a_{i1}, \dots a_{i4}.$ We also refer to   $D^A_i$ as the $A$-configuration for $\tilde{X}_i.$

Moreover, let $D^{B}_i$ be the search tree with root  $q_{t_i}$  and left subtree being 
a  sequential search tree for (some permutation of)
$a_{i1}, \dots a_{i4}.$
We say that   $D^B_i$ is  the $B$-configuration for $\tilde{X}_i.$
See Figure \ref{fig:Config-LinearStrat} (b)-(c).

\begin{definition} \label{defi: extension}
Given two search trees $T_1, T_2,$ the extension of $T_1$ with $T_2$ is the search tree obtained by appending the root of $T_2$ to the leftmost
leaf of $T_1.$ 
The  extension of $T_1$ with $T_2$ is a new search tree that ``acts'' like
 $T_1$ and in case of all NO answers continues following the strategy represented by $T_2.$
\end{definition}

\begin{definition} \label{defi:realization}
A realization (of $\Pi$) with respect to ${\cal Y} \subseteq {\cal X}$ is a search tree for $(T, w)$ defined recursively as follows:
\footnote{For sake of definiteness we set $\pi_{m+n+1} = \emptyset$ and the realization of $\pi_{n+m+1}$ w.r.t. $\cal Y$ to be the empty tree.}
For each $i=1, \dots, n+m, $ a realization of $\pi_i \, \pi_{i+1} \dots \pi_{n+m}$ is an extension of 
the realization of   $\pi_{i+1} \dots \pi_{n+m}$ with another tree $T'$ chosen according to the following two cases:

\noindent
{\em Case 1.} If $\pi_i=u_j,$ for some $j = 1, \dots, n,$ then $T'$ is  
a (possibly empty) sequential search tree for the  leaves of $T$ that are associated with $u_{j}$ and are
 not queried in the realization of $\pi_{i+1} \dots, \pi_{n+m}.$

\noindent
{\em Case 2.} If $\pi_i= X_j,$ for some $j=1,\dots,m,$ then $T'$ is either $D_j^B$ or $D_j^A$ according as $X_j \in {\cal Y}$ or not.
\end{definition}

We denote by $D^A$ the realization of $\Pi$ w.r.t. the empty family, i.e.,  ${\cal Y} = \emptyset.$ 
Figure \ref{fig:realizations} shows some of the realizations for 
the Example 1 above.

We are going to  set the weights in such a way that every optimal solution is a realization of $\Pi$ w.r.t.
some ${\cal Y} \subseteq {\cal X}$ (our Lemma \ref{lemma:optimal_structure}). Moreover, such  weights will  
allow to discriminate between the cost of solutions that are realizations w.r.t. to an exact cover for the X3C instance 
and the cost of any  other realization of $\Pi$. 
Let $D^*$ be an optimal search tree and ${\cal Y}$ be such that $D^*$ is a realization of $\Pi$ w.r.t. ${\cal Y}.$\footnote{The existence 
of such a ${\cal Y}$ will be guaranteed by Lemma \ref{lemma:optimal_structure}.}
In addition, for each $u \in U$ 
 define $W_u = \sum_{\ell : X_{\ell} \prec u}  w(\tilde{X}_{\ell})$.
 It is not hard to see that  the difference between the cost of $D^A$ 
 and $D^*$ can be expressed as follows:
\begin{equation} \label{eq:diff_cost}
cost(D^A)-cost(D^*) = \sum_{ X_i \in {\cal Y}} \left( w(t_i)  
- (W_{u_{i \, 1}}+ W_{u_{i \, 2}} + W_{u_{i \, 3}})
-\sum_{j=1}^3 d^{A}_{B}(q_{s_{i\, j}})  w(u_{i\,j}) \right), 
\end{equation}
where $d^{A}_{B}(q_{s_{i\,j}})$ is the difference between the level of the node $q_{s_{i\,j}}$ in  $D^*$ and the level $q_{s_{i\,j}}$ in a
 realization of $\Pi$ w.r.t. ${\cal Y} \setminus \{X_i\}.$  To see this, imagine to turn 
$D^A$ into $D^*$ one step at a time. Each step being the changing of configuration from $A$ to $B$ for a set of leaves $\tilde{X}_i$  such that
$X_i \in {\cal Y}.$ Such a step implies: 
(a) moving the question $q_{s_{ij}}$ exactly $d^A_B(q_{s_{ij}})$ levels down, so increasing the cost by 
$ d^{A}_{B}(q_{s_{i\, j}})  w(u_{i\,j});$
(b) because of (a) all the questions that were below the level where  $q_{s_{ij}}$ is moved, are also moved down one level.
This additional increase in cost  is accounted for by the $W_{u_{i\, j}}$'s;
(c)  moving one level up the question about $t_i,$ so gaining cost $w(t_i).$
 
We will define the weight of $t_i$ in order to:  compensate the increase in cost (a)-(b) due to the relocation of $q_{s_{i\, j}};$ 
and  to provide some additional 
gain only when ${\cal Y}$ is an exact cover. 
In general, the value of $d^A_B(q_{s_{i\,j}})$ depends on the structure of  the realization for ${\cal Y}\setminus 
\{X_i\};$ in particular, on  the length of the  sequential search trees for the leaves associated to $u_{\kappa}$'s, that appear in $\Pi$ between 
$X_i$ and $u_{i\, j}.$ However, when 
${\cal Y}$ is an exact cover, each such sequential search tree has length one. 
A moment's reflection shows that in this case $d^A_B( q_{s_{i\,j}}) = \gamma(i, j),$ where, for each $i = 1, \dots, m$ and $j=1,2,3,$ we define 
$$\gamma(i, j) = j - 5 + | \{ u_{\kappa} : u_{i \, j} \prec u_{\kappa} \prec X_i \} | + 5 \cdot | \{ X_{\kappa} : u_{i \, j} \prec X_{\kappa} \preceq X_i \} |$$

\remove{
To see this, assuming that ${\cal Y}$ is an exact cover,  let us compute  $d^A_B( q_{s_{i\,,j}})$  as follows: 
 $q_{s_{i\,j}}$ is exactly  $3 - j$ levels below  the deepest node of the $B$-configuration for $\tilde{X}_i$. 
Then the remaining distance can be computed by traversing $\Pi$ from $X_{i}$   to $u_{i\,j}$ and 
adding 1 for each $u_{\kappa}$ encountered (accounting for the sequential search tree for the only leaf associated to $u_{\kappa}$) 
and adding $3$ for each $X_{\kappa}$ encountered (accounting for the left path of the $A$ or $B$ configuration taking care of $\tilde{X}_{\kappa}$).  
}

To see this,  assume that ${\cal Y}$ is an exact cover.
Let $D'$ be the realization for ${\cal Y} \setminus X_i,$ and  $\ell$
be the level of the root of the $A$-configuration  for $\tilde{X}_i$ in $D'$.
The node $q_{s_{i\,j}}$ is at level $\ell+(5-j)$ in $D'$.
In $D^*,$ the root of the $B$-configuration  for $\tilde{X}_i$ is
also at level $\ell.$ Also, in $D^*,$  between level $\ell$ and the level of $q_{s_{i\,j}},$ there are only nodes
associated with elements of some $\pi_k$ s.t. $u_{ij} \prec \pi_{\kappa} \preceq X_j.$
Precisely, there is $1$ level per each $u_{\kappa}$ s.t. $u_{i j} \prec u_{\kappa} \prec 
X_{i}$ (corresponding to the sequential search tree for the only leaf associated 
with $u_{\kappa}$); and $5$ levels per each  ${X}_{\kappa}$ s.t. 
$u_{i j} \prec X_{\kappa} \preceq X_{i}$ (corresponding to the left path of the $A$ or $B$-configuration 
for $\tilde{X}_{\kappa}$). In total, the difference between the levels of $q_{s_{ij}}$ in 
$D'$ and $D^*$ is exactly $\gamma(i,j).$
\remove{start at level $\ell$ and  traverse $\Pi$ from $X_{i}$ to $u_{i\,j}$
adding 1 for each $u_{\kappa}$ encountered (accounting for the sequential
search tree for the only leaf associated to $u_{\kappa}$) and adding $2$
for each $X_{\kappa}$ encountered (accounting for the left path of the $A$
or $B$ configuration taking care of $\tilde{X}_{\kappa}$). By taking the
difference
between the level of $q_{s_{i\,j}}$ in $D^*$ and $D$ we obtain
$\gamma(i,j)$. }

Note that $\gamma(i,j)$ is still well defined
even if there is not an exact cover ${\cal Y} \subseteq {\cal X}$.
This quantity will be used  to define $w(t_i).$

We are now  ready to provide the precise  definition of the weight function $w.$ 
We start with $w(u_1) = 1.$ Then, we fix the remaining weights inductively, using the sequence $\Pi$ in the following way:  
let $i > 1$ and  assume that for each $i' < i$ the weights of all leaves associated with $\pi_{i'}$ have been fixed\footnote{By
the leaves associated with $\pi_{i'}$ we mean the leaves in $\tilde{X}_{j}$, if $\pi_i = X_j$ for some $X_j \in {\cal X}$, 
or the leaves associated with $u$ if $\pi_{i'} = u$ for some $u \in U.$}.
We now  proceed according to the following two cases:

\noindent
{\em Case 1.} $\pi_i = u_j,$ for some $j \in \{1, \dots, n\}.$  Then,  we set $w(u_j) = 1+ 6\max\{|T|^3 w(u_{j-1}), W_{u_j}\},$ 
where $|T|$ denotes the number of nodes of $T.$

\noindent
{\em Case 2.} $\pi_i = X_j,$ for some $j \in \{1,\dots, m\}.$ Note that in this case the weights of the leaves 
$s_{j \, 1}, s_{j \, 2}, s_{j \, 3}$ have  already been fixed, respectively to $w(u_{j \, 1}), w(u_{j \, 2}),$ and $w(u_{j \, 3}).$
This is because we fix the weights following the sequence $\Pi$ and  we have 
$u_{j\, 1} \prec u_{j\, 2} \prec u_{j\, 3} \prec X_j.$
 In order to define the weights of the 
remaining elements in $\tilde{X}_j$ we set  
$w(a_{j1})  = \dots = w(a_{j4})=  W_{u_{j\, 1}}+ W_{u_{j\, 2}} + W_{u_{j\, 3}} +  \sum_{\kappa = 1}^3  \gamma(j, \kappa) w(u_{j\, \kappa})$. Finally, we set 
$w(t_j) = w(a_{j1}) + w(X_j)/2.$

\begin{remark} \label{remark:weight_encoding}
For each $i=1, \dots, n+m,$ let  $w(\pi_i)$ denote the total weight of the leaves associated with $\pi_i.$ 
It is not hard to see that  
$w(\pi_i) = O(|T|^{3 i}).$ 
Therefore we have that the maximum weight is not larger than $w(\pi_{m+n})= O(|T|^{3(m+n)}).$
It follows that we can encode all the weights using $O(3|T|(n+m) \log |T|)$ bits, hence the size of the instance 
$(T,w)$ is polynomial in the size of the X3C instance $\mathbb{I} = (U, {\cal X}).$
\end{remark}


Since $t_m$ is the heaviest leaf, one can show that in an optimal search tree $D^*$ the root can only be $q_{t_m}$
or $q_{r_m}.$ For otherwise moving one of these questions closer to the root of $D^*$ results in a tree with 
smaller cost, violating the optimality of $D^*.$ Moreover, by a similar ``exchange'' argument it follows that 
if $q_{r_m}$ is the root of $D^*$ then the right subtree must coincide with a sequential  search tree for 
$t_m, s_{m1}, s_{m2}, s_{m3}$ and the left subtree of  $q_{r_m}$ must be a sequential tree for 
$a_{m1},\dots, a_{m4}.$
Therefore the top levels of $D^*$ coincide either with $D^A_{m}$ or with $D^B_{m},$ or equivalently they are a realization 
of $\pi_{m+n}.$ Repeating  the same argument on the remaining part of $D^*$ we have the following (the complete proof is in 
appendix):




\begin{lemma} \label{lemma:optimal_structure}
Any optimal search tree for the instance $(T, w)$ is a realization of $\Pi$ w.r.t. some ${\cal Y} \subseteq {\cal X}.$
\end{lemma}
\remove{
The proof consists to show  that the upper part of the optimal tree is a realization for 
$\pi_{i+1} \ldots \pi_{m+n}$ for some $i$ and then to argue
through a case analysis that the upper part of the optimal tree is also a realization for  $\pi_{i} \ldots \pi_{m+n}$.
For the complete proof can be found in the appendix.
}

Recall now the definition of the search tree $D^A.$ Let $D^*$ be an optimal  search tree for $(T, w).$
Let ${\cal Y} \subseteq {\cal X}$ be such that $D^*$ is a realization of $\Pi$ w.r.t. $\cal Y.$ 
Equation (\ref{eq:diff_cost}) and the definition of $w(t_i)$ yield

\begin{equation} \label{eq:Delta_cost}
cost(D^A)-cost(D^*) = \sum_{ X_i \in {\cal Y}} \left ( \frac{w(X_i)}{2} +   
\sum_{j=1}^3 \left( \gamma(i,j) - d^{A}_{B}(q_{s_{i\,j}}) \right) w(u_{i\,j})   \right)
= \sum_{j=1}^n \mathop{\sum_{X_i \in {\cal Y}}}_{u_j \in X_i} \left( \frac{w(u_j)}{2} + 
\Gamma(i,j) w(u_j) \right),
\end{equation}
where $\Gamma(i,j) = \gamma(i,\kappa) - d^A_B(q_{s_{i \, \kappa}}),$ and $\kappa  \in \{1,2,3\}$ is such that $s_{i\, \kappa} = u_j.$

\remove{
By the definition of $\gamma(\cdot, \cdot)$ and $d(\cdot)$,  
if each $u \in U$ appears  exactly in one set of ${\cal Y}$ then (\ref{eq:Delta_cost}) is equal  to $\sum_{j=1}^n \frac{w(u_j)}{2}.$ 
Conversely, let $1 \leq j \leq n$ be such that $u_{j'}$ appears exactly once in $\cal Y,$ for each $j' > j$ and  $u_j$ does not appear in any set in $\cal Y.$ a
   
Then  ((\ref{eq:Delta_cost}) is not larger than  
$\sum_{j=1}^n \frac{w(u_j)}{2} - w(u_j)/2 +  \sum_{j=1}^j-1 \mathop{\sum_{X_i \in {\cal Y}}}_{u_j \in X_i} \frac{w(u_j)}{2} + 
\Gamma(i,j) w(u_j),$ which is smaller than $\sum_{j=1}^n \frac{w(u_j)}{2}$ even if all the $Delta(\cdot,\cdot)$'s where equal to $|T|,$ since 
$w(u_{j'}) \leq |T|^3 w(u_{j}),$ for each $j' < j.$
On the other hand if there is more than one set in $\cal Y$ which includes $u_j,$ then at most one of the $\Gamma(i,j)$'s is zero while the others are 
negative. Therefore since, by the previous argument,  the remaining $u_{j'}$ with $j' < j$ can only provide  total weight at most $w(u_{j})$ we have again that 
$(\ref{eq:Delta_cost})$ becomes smaller than $\sum_{j=1}^n \frac{w(u_j)}{2}.$
}

By definition, if for each $j = 1, \dots, n,$ there exists exactly one $X_i \in {\cal Y}$ such that $u_j \in X_i,$
then we have $\Gamma(i,j)=0.$ Therefore,  equation (\ref{eq:Delta_cost}) evaluates 
exactly to $\sum_{j=1}^n \frac{w(u_j)}{2}.$
Conversely, we can prove that this never happens when for some $1\leq j\leq n,$ 
$u_j$ appears in none or in more than one of the sets in ${\cal Y}.$
For this we use the exponential (in $|T|$) growth of  the weights $w(u_j)$ and the fact that in such case the inner sum of the
last expression in (\ref{eq:Delta_cost}) is non-positive.
\remove{Conversely, suppose $j^*$ is the maximum index such that $u_{j^*}$ appears in none or more than one of the sets in 
${\cal Y}.$ Therefore, there exists at most one $X_i$ for which $\Gamma(i,j^*)=0.$
For the others (if any)  the $\Gamma(\cdot,j^*)$ is negative. Therefore, the total contribution of the inner sum of the right-hand side of
 (\ref{eq:Delta_cost}) for $j=j^*$ becomes non positive.
 Moreover, because of the difference among the weights, the sum of the contributions for $j < j^*$ 
is less than $w(u_{j^*})/2.$}
 In conclusion we have the following result, whose 
complete proof is in appendix.

\begin{lemma} \label{lemma:key_2}
Let $D^*$ be an optimal search tree for $(T, w).$
Let ${\cal Y} \subseteq {\cal X}$ be such that $D^*$ is a realization of $\Pi$ w.r.t. $\cal Y.$
We have that $cost(D^*) \leq  cost(D^A) - \frac{1}{2}\sum_{u \in  U} w(u)$  if and only if  ${\cal Y}$ is a solution for the 
X3C instance $\mathbb{I} = (U, {\cal X}).$
\end{lemma}

The ${\cal NP}$-Completeness of 3-bounded X3C \cite{Gar79a},  Remark \ref{remark:weight_encoding}, 
and Lemma \ref{lemma:key_2} imply the following. 

\begin{theorem} \label{thm:main}
The search tree problem  is ${\cal NP}$-Complete in the class of trees of diameter at most $4.$ \commento{I changed the statement .}
\end{theorem}
Note that this result is tight. In fact,  for trees of diameter at most $3$ the problem is 
polynomially solvable, e.g.,  via dynamic programming  (see Appendix).

\medskip

\noindent
{\bf ${\cal NP}$-Completeness for bounded-degree instances.}
We can adapt our proof to show that the search tree problem is ${\cal NP}$-Complete also
for bounded-degree trees. For that, we modify the input tree as follows.
We partition the subsets of ${\cal X}$ so that sets that are adjacent in $\Pi$ are put together. 
For the instance in the Example 1 the corresponding partition would be $\{\{X_1\}, \{X_2\}, \{X_3, X_4\}\}.$ 

Let $\mathbb{Z}=\{{\cal Z}_1,\ldots, {\cal Z}_p\}$ be the partition obtained from the input instance $(U, {\cal X}).$ 
Recall the definitions of the subtrees $T_j$ and the leaves $a_{j1},\dots, a_{j4}$ ($j=1, \dots, m$)
given for the construction of the tree $T.$ We now create a new tree $T^b$ as follows.
For each $i=1, \dots, p,$ in $T^b$ there is a  subtree $H_i$  that corresponds to the  element ${\cal Z}_i \in \mathbb{Z}.$ 
$H_i$ has root $h_i.$ For each  $j$ such that $X_j \in {\cal Z}_i$ we make the root of $T_j,$ i.e., $r_j,$ and the leaves 
$a_{j1}, \dots, a_{j4}$ children of  $h_i.$ Finally, we create nodes $z_1, \dots, z_p$ and 
make $h_1$ a child of $z_1$ and for $i=2, \dots, p$ we make $z_{i-1}$ and $h_i$ children of $z_i.$ See Fig. \ref{fig:tree_b} for
 the tree $T^b$ corresponding 
to the instance in Example 1.

The fact that in $\Pi$ there are no more than three elements of ${\cal X}$ which appear consecutively, implies that any ${\cal Z}_i$
contains at most three elements. This  gives that the maximum degree in $T^b$ is at most
$16$.

Regarding the weight function, we extend to $T^b$  the weight function defined for the tree $T$ by setting  $w(h_i) = w(z_i) = 0,$ for each $i=1, \dots, p$ 
and leaving the other weights as before. 

It turns out  that  Lemma \ref{lemma:optimal_structure} still holds for the new  instance $(T^b, w).$ 
In fact, in each subtree $H_i$ the structure of the instance is 
exactly the same as in the tree $T,$ so one can prove that any optimal solution for such subinstance is a realization of the corresponding subsequence of 
$\Pi.$ Moreover, because of the way we partitioned $\cal X,$ and the weight function $w,$ it follows that the smallest weight of an  
$a_{jk}$ in ${\cal Z}_i$ is bigger than  the total weight of the leaves in ${\cal Z}_1, \dots, {\cal Z}_{i-1}.$ This is enough to enforce the order 
of a realization of $\Pi,$ i.e., that the leaves $t_j, a_{j1},\dots, a_{j4}$ are queried before the leaves in ${\cal Z}_1, \dots, {\cal Z}_{i-1}.$
We have proved the following (a formal proof is in the appendix).
\begin{lemma} \label{lemma:bounded_optimal_structure}
Any optimal search tree for the instance $(T^b, w)$ is a realization of $\Pi$ w.r.t. some ${\cal Y} \subseteq {\cal X}.$
\end{lemma}
By using this  lemma together with Lemma \ref{lemma:key_2} we have that Theorem \ref{thm:main}
holds also for bounded-degree instances of the tree search problem.


\section{Approximation Algorithms}
We need to introduce some notation. For any forest $F$ of rooted trees and node $j \in F$, we denote by $F_j$ the subtree of $F$ composed by $j$ and all of its descendants.  We  denote the root of a tree $T$ by $r(T)$,
$\delta(u)$ denotes the number of children of $u$ and $c_i(u)$ is used to denote the $i$th  child of $u$
according to some arbitrarily fixed order. 
The following operation will be  useful  for modifying search trees:
%
	Given a search tree $D$ and a node $u \in D$, a \emph{left deletion} of $u$ is the operation that transforms $D$ into a new search tree by removing both $u$ and its left subtree from $D$ and, then, by
connecting the right subtree of $u$ to the parent of $u$ (if it exists). A \emph{right deletion} is analogously defined.

Given a search tree $D$ for $T$, we use $l_u$ to denote the leaf of $D$ assigned to node $u$ of $T$.
\subsection{The natural greedy algorithm attains $2$-approximation}
		Consider a search tree $D$ for $T$. Notice that when we follow a path from the root of $D$ to one of its leaves, we reduce the search space (eliminate part of $T$) whenever we visit a new node. Therefore, we can associate with each node of $D$ the subtree of $T$ which may still contain the node we search for. Notice that the tree $T'$ associated with node $v \in D$ is exactly the one induced by the nodes of $T$ that correspond to the leaves of $D_v$, hence $w(T') = w(D_v)$.
	E.g., in Fig.\ \ref{fig:prob-defi} the node $<f>$ in $D$ is associated with $T_{d}.$ 

	We can transform a search tree $D$ for $T$ into a search tree $D'$ for an arbitrary subtree $T'$ of $T$. This search tree $D'$ is computed by taking each node $v \in D$ assigned to a node $A(v)$ in $T - T'$ and applying a left deletion if $A(v)$ is an ancestor of $r(T')$ or a right deletion otherwise. The important property of this construction is that the path $r(D') \leadsto l_x$,
	for every $x \in T'$, is exactly the subpath obtained by removing all queries to nodes in $T - T'$ from $r(D) \leadsto l_x$. The next lemma formalizes this discussion:
	
	\begin{lemma} \label{subSS}
		Consider a tree $T$ and a search tree $D$ for it. Let $T'$ be a subtree of $T$. Then there is a search tree $D'$ for $T'$ such that $\dist{r(D')}{l_x}{D'} = \dist{r(D)}{l_x}{D} - n_x$, where $n_x$ is the number of nodes in the path $r(D) \leadsto l_x$ assigned to nodes in $T - T'$.
	\end{lemma}

We show that the natural greedy algorithm guarantees an approximation factor of $2$. The algorithm 
can be formulated in two sentences. (1) Let $x$ be a node such that $| w(T_x) - w(T \setminus T_x)|$ is  minimized. Set $A(r(D)) = x.$ (2) Construct the right and left subtree of $D$ by recursively applying the algorithm to $T_x$ and  $T \setminus T_x,$ respectively.  

In order to prove that this algorithm results in a $2$-approximation, we show that any search tree $D^*$ can be turned into the greedy
search tree $D$ while the cost increases by at most $cost(D^*).$

The proof is by induction on the number of nodes $n$ of the input tree $T.$ For the basic case  $n=1$ there is nothing to show. 
Assume that the claim holds for any tree with at most $n-1$ nodes. In order to prove it true for $T$ we proceed in two steps. 

Let $x$ be the node queried at the root of $D.$ Also let $D^*_0$ (resp. $D_0$) and $D^*_1$ (resp. $D_1$) 
be the search tree for $T_x$ and $T \setminus T_x$ obtained from $D^*$ (resp. $D$)  via Lemma \ref{subSS}. (a) Construct a search tree $D'$ with $A(r(D')) = x$ and the left and right subtree being $D^*_1$ and $D^*_0$ 
respectively. It is not hard to see that $D'$ is a legal search tree. (b) Use the induction hypothesis for turning $D^*_0$ and $D^*_1$ into 
$D_0$ and $D_1$ respectively. It is straightforward to see that the transformation results in the tree $D.$

\begin{lemma} \label{lemma:1_ESA}
We have $cost(D') \leq cost(D^*) + w(T)/2.$
\end{lemma}

\begin{proof}[Proof sketch]
Let $x$ and $x^*$ be the nodes queried at the root of $D^\prime$ and $D^*$, respectively. W.l.o.g. we assume $x \neq x^*$, as otherwise the lemma trivially holds. We can also assume that $x^*$ is a node from $T_x$, because the opposite case is analyzed analogously.

We shall first analyze the case $w(T_x) \leq w(T - T_x),$ i.e.,  $w(T_x) \leq w(T)/2$. As any path from $r(D^*)$ to a leaf in $D^*$ contains $r(D^*)$ and $T - T_x$ does not contain $x^*$, Lemma~\ref{subSS} states that the depth of any leaf in $D^*_1$ is at least by one smaller than it is in $D^*$. The lemma also implies that the depth of any leaf in $D^*_0$ is not greater than it is in $D^*$. So we have
$$cost(D^\prime) = 
w(T) + cost(D_0^*) + cost(D_1^*) 
	 \leq 
	w(T) + cost(D^*) - w(T - T_x) \leq cost(D^*) + w(T)/{2}.
$$
The case  $w(T_x) > w(T - T_x)$ requires a  more involved analysis and we defer it to the appendix due to the space limitations.
\end{proof}
It follows that the cost of $D$ can be bounded from above by 
$$cost(D) = w(T) + cost(D_0) + cost(D_1) \leq w(T) +2 cost(D^*_0) + 2cost(D^*_1) = 
2cost(D') - w(T) \leq 2cost(D^*).$$
The first inequality follows from the induction hypothesis and the second one is due to Lemma \ref{lemma:1_ESA}.

We have proven the following result.

\begin{theorem}
The greedy strategy is a polynomial  $2$-approximation algorithm for the tree search problem. 
\end{theorem}

\subsection{An FPTAS for Searching in Bounded-Degree Trees}
\label{sec:search}
	We now present an FPTAS for searching in trees with bounded degree. First, we devise a dynamic programming algorithm whose running time is exponential in the height of optimal search trees. Then we essentially argue that the height of optimal search trees  is $O(\Delta(T) \cdot (\log w(T) + \log n))$, thus the previous algorithm has a pseudo-polynomial running time. Finally, we employ a standard scaling technique to obtain an FPTAS.

\remove{
 We also extend the set difference operation to trees: given trees $T^1 = (V^1, E^1)$ and $T^2 = (V^2, E^2)$, $T^1 - T^2$ is the forest of $T^1$ induced by the nodes $V^1 - V^2$.}



We often construct a search trees starting with its `left part'. In order to formally describe such constructions, we define a \emph{left path} as an ordered path where every node has only a left child. In addition, the \emph{left path} of an ordered tree $T$ is defined as the ordered path we obtain when we traverse $T$ by only going to the left child, until we reach a node which does not have a left child.

\medskip \noindent
{\bf A dynamic programming algorithm.} 
In order to find an optimal search tree in an efficient way, we need to define a family of auxiliary problems denoted by $\mathcal{P}^B(F,P)$. In the following paragraphs we describe the essential structures needed in these subproblems and then we show how to use
the subproblems to find an optimal search tree.

\remove{	In order to find an optimal search tree in an efficient way, we need to define a family of auxiliary problems denoted by $\mathcal{P}^B(F,P)$. In the following paragraphs we describe the essential structures needed in these subproblems and then we show that we can use a problem $\mathcal{P}^B$ to find an optimal search tree. }

	First we  introduce the concept of an \emph{extended search tree}, which is basically a search
tree with some extra nodes that have not been associated with a query yet (unassigned nodes) and some other
nodes that cannot be associated with a query (blocked nodes). 


	\begin{definition}
	An \emph{extended search tree} (EST) for a forest $F=(V,E)$ is a triple $D = (N, E', A)$, where $N$ and $E'$ are the nodes and edges of an ordered binary tree and the assignment $A :  N \rightarrow V \cup \{\bl,\un\}$ simultaneously satisfy the following properties:
	\begin{enumerate}
		\item[(a)] For every node $v$ of $F$, $D$ contains both a leaf $\ell$ and an internal node $u$ such that 
		$A(\ell) = A(u) = v$; 
		\item[(b)] $\forall u,v \in D$, with $A(u),A(v) \in F$, the following holds:
		 If  $v$ is in the right subtree of $u$ then $A(v) \in F_{A(u)}$. 
		 If $v$ is in left subtree of $u$ then $A(v) \notin F_{A(u)}$;
		 \item[(c)] If $u$ is a node in $D$ with $A(u) \in \{\bl,\un\}$, then $u$ does not have a right child.
	\end{enumerate}
	\end{definition}

	If we drop (c) and also the requirement regarding internal nodes in (a) we have the definition of a search tree for $F$. The cost of an EST $D$ for $F$ is analogous to the cost of a search tree and is given by $cost(D) = \sum d(r(D), u) w(A(u))$, where the summation is taken over all leaves $u \in D$ for which $A(u) \in F$.
	
	At this point we establish a correspondence between optimal EST's and optimal search trees. Given an EST $D$ for a tree $T$, we can apply a 
	left deletion to the internal node of $D$ assigned to $r(T)$ and right deletions to all nodes of $D$ that are blocked or unassigned, getting a search tree $D'$ of cost $cost(D') \le cost(D) - w(r(T))$. Conversely, we can add a node assigned to $r(T)$ to a search tree $D'$ and get an EST $D$ such that $cost(D) \le cost(D') + w(r(T))$. Employing these observations we can prove the following lemma:
	
	\begin{lemma} \label{ESTtoST}
	Any  optimal EST for a tree $T$ can be converted into an optimal search tree for $T$ (in linear time). In addition, the existence of an optimal search tree of height $h$ implies the existence of an optimal EST of height $h + 1$.
	\end{lemma}

	So we can focus on obtaining optimal EST's. First, we introduce concepts which serve as a building blocks for EST's. A \emph{partial left path} (PLP) is a left path where every node is assigned (via a function $A$) to either $blocked$ or $unassigned$. Now consider an EST $D$ and let $L = \{l_1, \ldots, l_{|L|}\}$ be its left path. We say that $D$ is \emph{compatible} with a PLP $P = \{p_1, \ldots, p_{|P|}\}$ if $|P|=|L|$ and $A(p_i) = \textrm{\emph{blocked}}$ implies $A(l_i) = \textrm{\emph{blocked}}$.
The tree in Figure \ref{fig:case1}.(c) is compatible with the path of 	Figure \ref{fig:case1}.(b).

	This definition of compatibility implies a natural one to one correspondence between nodes of $L$ and $P$. Therefore, without ambiguity, we can use $p_i$ when referring to node $l_i$ and vice versa.
	
	Now we can introduce our subproblem $\mathcal{P}^B$. First, fix a tree $T$ with $n$ nodes and a weight function $w$.  Given a forest $F = \{T_{c_1(u)}, T_{c_2(u)} \ldots, T_{c_f(u)}\}$, a PLP $P$ and an integer $B$, 
the problem $\mathcal{P}^B(F,P)$ consists of  finding 
an EST for $F$ with minimum cost among those EST's for $F$ that are compatible with $P$ and have height at most $B$.
We  shall note  that $F$ is not a general subforest of $T$, but one consisting of  subtrees rooted
at  the first $f$ children of some node $u \in T$, for some $ 1 \leq f \leq \delta(u)$. 	

	Notice that if $P$ is a PLP where all nodes are unassigned and $P$ and $B$ are sufficiently large, then $\mathcal{P}^B(T, P)$ gives an optimal EST for $T$.

\medskip \noindent
{\bf Algorithm for $\mathcal{P}^B(F,P)$.}  We have a base case and also two other cases depending on the structure of $F$. In all these cases, although not explicitly stated, if $P$ does not contain unassigned nodes then the algorithm returns `not feasible'. If during its execution the  algorithm encounters a `not feasible' subproblem it ignores this choice in the enumeration. 

\medskip \noindent {\bf \em Base case:} {\em $F$ has only one node $u$.} In this case, the optimal solution for $\mathcal{P}^B(F,P)$ is obtained from $P$ by assigning its first unassigned node, say $p_i$, to $u$ and then adding a leaf assigned to $u$ as a right child of $p_i$.
Its cost is $i \cdot w(u)$.	

\medskip \noindent {\bf \em Case 1:} {\em $F$ is a forest $\{T_{c_1(u)}, \ldots, T_{c_f(u)}\}$.} The idea of the algorithm is to decompose the problem into subproblems for the forests $T_{c_f(u)}$ and $F \setminus T_{c_f(u)}$. For that, it needs to select which nodes of $P$ will be assigned to each of these forests. 

	The algorithm considers all possible bipartitions  of the unassigned nodes of $P$ and for each 
	bipartition $\mathcal{U} = (U^f, U^o)$ it computes an EST $D^{\mathcal{U}}$ for $F$ compatible with $P$. At the end, the algorithm returns the tree $D^{\mathcal{U}}$ with smallest cost. The EST $D^{\mathcal{U}}$ is constructed as follows:
	
	\begin{enumerate}
		\item Let $P^f$ be the PLP constructed by starting with $P$ and then setting all the nodes in $U^o$ as blocked (Figure \ref{fig:case2}.b). Similarly, let $P^o$ be the PLP constructed by starting with $P$ and setting all nodes in $U^f$ as blocked. 
		 Let $D^f$ and $D^o$ be optimal solutions for $\mathcal{P}^B(T_{c_f(u)},P^f)$ and $\mathcal{P}^B(F \setminus T_{c_f(u)}, P^o)$, respectively (Figure \ref{fig:case2}.c).
		
		\item The EST $D^{\mathcal{U}}$ is computed by taking the `union' of $D^f$ and $D^o$ (Figure \ref{fig:case2}.d). More formally, the `union' operation consists of starting with the path $P$ and then replacing: (i) every  node in $P \cap U^f$  by the corresponding node in the left path of $D^f$ and its right subtree; (ii) every  node in $P \cap U^o$  by the corresponding node in the left path of $D^o$ and its right subtree.
	\end{enumerate}

   	Notice that the height of every EST $D^{\mathcal{U}}$ is at most $B$; this implies that the algorithm returns a feasible solution for $\mathcal{P}^B(F, P)$. Also, the cost of $D^{\mathcal{U}}$ is given by $OPT(\mathcal{P}^B(T_{c_f(u)},P^f))+OPT(\mathcal{P}^B(F \setminus T_{c_f(u)},P^o))$.

		The optimality of the above procedure relies on the fact we can build an EST $\bar{D}^f$ for $T_{c_f(u)}$ by starting from an optimal solution $D^*$ for $\mathcal{P}^B(F,P)$ and performing the following operation at each node $v$ of its left path: (i) if $v$ is unassigned we assign it as blocked; (ii) if $v$ is assigned to a node in $F \setminus T_{c_f(u)}$ we assign it as blocked and remove its right subtree. We can construct an EST $\bar{D}^o$ for $F \setminus T_{c_f(u)}$ analogously. Notice that $cost(\bar{D}^f) + cost(\bar{D}^o) = cost(D^*)$. The proof is then completed by noticing that, for a particular choice of $\mathcal{U}$, $\bar{D}^f$ and $\bar{D}^o$ are feasible for $\mathcal{P}^B(T_{c_f(u)},P^f)$ and $\mathcal{P}^B(F \setminus T_{c_f(u)}, P^o)$, so the solution
returned by the above algorithm costs at most 
$OPT(\mathcal{P}^B(T_{c_f(u)},P^f))+OPT(\mathcal{P}^B(F \setminus T_{c_f(u)},P^o)) \le cost(D^*)$.

\medskip \noindent {\bf \em Case 2:} {\em $F$ is a tree $T_v$.} Let $p_i$ be an unassigned node of $P$ and let $t$ be an integer in the interval $[i+1, B]$. The algorithm considers all possibilities for $p_i$ and $t$ and computes an EST $D^{i,t}$ for $T_v$ of smallest cost satisfying the following: (i) $D^{i,t}$ is compatible with $P$; (ii) its height is at most $B$; (iii) the  node of the left path of $D^{i,t}$ corresponding to $p_i$ is assigned to $v$; (iv) the leaf of $D^{i,t}$ assigned to $v$ is located at level $t$. The algorithm then returns the  tree $D^{i,t}$ with minimum cost. 

In order to  compute $D^{i,t}$ the algorithm executes the following steps: 
\begin{enumerate}
		\item Let $P^i$ be the subpath of $P$ that starts at the first node of $P$ and ends at $p_i$. Let $P^{i,t}$ be a left path obtained by appending $t - i$ unassigned nodes to $P^i$ and assigning $p_i$ as blocked (Figure \ref{fig:case1}.b).
 Compute an optimal solution $D'$  for $\mathcal{P}^B(\{T_{c_1(v)}, T_{c_2(v)}, \ldots, T_{c_{\delta(v)}(v)}\}, P^{i,t})$. 

		\item Let $p'_i$ be the node of $D'$ corresponding to $p_i$ and let $y'$ be the last node of the left path of $D'$ (Figure \ref{fig:case1}.c). The tree $D^{i,t}$ is constructed by modifying $D'$ as follows (Figure \ref{fig:case1}.d): make the left subtree of $p'_i$ becomes its right subtree; assign $p'_i$ to $v$; add a leaf assigned to $v$ as the left child of $y'$; finally, as a technical detail, add some blocked nodes to extend the left path of this structure until the left path has the same size of $P$. \label{step3}
\end{enumerate}
	

		It follows from properties (i) and (ii) of the trees $D^{i,t}$'s that the above procedure returns a feasible solution for $\mathcal{P}^B(T_v, P)$. The proof of the optimality of this solution uses the same type of arguments as in Case 1 and is deferred to the appendix.

\medskip \noindent
{\bf Computational complexity.} Notice that it suffices to consider problems $\mathcal{P}^B(F, P)$'s where $|P| \le B$, since all others are infeasible. We claim that, by employing a Dynamic Programming strategy, we can compute all these problems in $O(n^2  2^{2B})$ time. First, there are $O(n 2^B)$ such problems; this follows from the fact that for each node $u$ in $T$ there are two possible forests $F$ considered in subproblems ($F = T_u$ or $F = \{T_{c_1(u')}, T_{c_2(u')}, \ldots, T_{c_f(u')} = T_u\}$, where $u$ is the $f$-th child of $u'$) and the fact there are $O(2^B)$ PLP's of size at most $B$. It is not difficult to see that each of these problems can be solved in $O( n + 2^B)$ time, so the claim holds.
 

\medskip \noindent
{\bf An upper bound on the height of optimal search trees.}
	We now argue that there is an optimal search tree for $(T, w)$ whose height is $O(\Delta(T) \cdot(\log w(T) + \log n))$.

	The following lemma is the core of our `geometric decrease' argument. 
It essentially states that we can cut a constant factor of the total weight of an optimal search tree  by going down a number of levels that only depends on  the maximum degree of $T$.

	\begin{lemma} \label{lemma:geoDec}
		Consider an instance $(T,w)$ for our search problem and let $D^*$ be an optimal search tree for it. 
		Fix  $0 \le \alpha < 1$ and an integer  $c > 3 (\Delta(T) + 1)/\alpha$. Then, for every node $v^* \in D^*$ with $\dist{r(D^*)}{v^*}{D^*} \ge c$ we have that $w(D^*_{v^*}) \le \alpha \cdot w(D^*)$.
	\end{lemma}
	
\begin{proof}[Proof sketch] (The full proof is deferred to the
appendix.) By means of contradiction assume the lemma does not hold
for some  $v^*$ satisfying its conditions. Let $\tilde{T}$ be the
tree associated with $v^*$, rooted at node $\tilde{r}$. Since by
hypothesis $\tilde{T}$ contains a large portion of the total weight
(greater than $\alpha \cdot w(D^*)$), we create the following search
tree $D'$ which makes sure parts of $\tilde{T}$ are queried closer to
$r(D')$: the root of $D'$ is assigned to $\tilde{r}$; the left tree of
$r(D')$ is a search tree for $T - T_{\tilde{r}}$ obtained via Lemma
\ref{subSS}; in the right tree of $r(D')$ we build a left path
containing nodes corresponding to queries for $c_1(\tilde{r}),
c_2(\tilde{r}), \ldots, c_{\delta(\tilde{r})}(\tilde{r})$, each
having as right subtree a search tree for the corresponding
$T_{c_i(\tilde{r})}$ obtained via Lemma \ref{subSS}. If $\bar{s}$ is
the number of nodes of $T - T_{\tilde{r}}$ queried in $r(D^*)
\leadsto v^*$, then Lemma \ref{subSS} implies that $D'$ saves at
least $\bar{s} - (\Delta(T) + 1)$ queries for each node in
$\tilde{T}$ when compared to $D^*$; this gives the expression
$cost(D') \le  cost(D^*) -  \bar{s} \cdot w(\tilde{T}) + (\Delta(T) +
1) w(T)$. Using the hypothesis on $c$ and $w(\tilde{T})$, this is
enough to reach the contradiction $cost(D') < cost(D^*)$ when
$\bar{s} \ge c/3$. The case when $\bar{s} < c/3$ is a little more
involved but uses a similar construction, only now the role of
$\tilde{r}$ is taken by a node inside $T_{\tilde{r}}$ in order to obtain
a more `balanced' search tree.
  \end{proof}

	Assume that the weight function $w$ is strictly positive (see Appendix \ref{app:FPTAS} for the general case). 
	Since $w$ is integral, employing Lemma  \ref{lemma:geoDec} repeatedly shows that $D^*$ 
	has height at most 
	$O(\Delta(T) \cdot (\log w(T) + \log n))$. 

\medskip \noindent
{\bf From the DP  algorithm to an FPTAS.} 
By Lemmas \ref{ESTtoST} and \ref{lemma:geoDec},  we can obtain an optimal search tree for $(T,w)$ by finding an optimal EST of height $B = O(\Delta(T) \cdot (\log w(T) + \log n))$ (via $\mathcal{P}^B$) and then converting it into an optimal search tree. Since we can employ the algorithm presented in the previous section to achieve this in $O\left( (n \cdot w(T))^{O(\Delta(T))}\right)$ time, we obtain a  pseudo-polynomial time algorithm for trees with bounded degree.
Furthermore, such an algorithm can  be transformed into an FPTAS by scaling and rounding the weights $w$, just as in the well-known FPTAS for the knapsack problem \cite{IbaKim75a} (see the appendix for details):	

 	\newcounter{thmFPTASc}
 	\setcounter{thmFPTASc}{\value{theorem}}	
	\begin{theorem} \label{FPTAS}
		Consider an instance $(T,w)$ to our search problem where $\Delta(T) = O(1)$. Then there is a 
		$poly(n \cdot w(T))$-time algorithm for computing an optimal search tree for $(T,w).$ 
In addition, there is a $poly(n/\epsilon)$-time algorithm for computing an $(1 + \epsilon)$-approximate search tree for $(T,w)$.
	\end{theorem}

\pagebreak


{\small
\bibliographystyle{abbrv}
\bibliography{bibliography}
}

\newpage

\appendix

\noindent{\normalfont\Large\bfseries Appendix}

\remove{ \section{The proof of the claim in remark \ref{remark:weight_encoding}: The magnitude of the weights.}
We shall show that $w(\pi_i) = O(|T|^{3 i}),$ for each $i=1, \dots, n+m.$

\begin{proof}
The claim is true for $i \leq 3,$ since 
in this case we have $\pi_i = u_i,$ and the  
definition gives us $w(u_i) = |T|^3w(u_{i-1}).$  

By definition, we have  $W_{u_j} \leq  w(u_j),$ for each $j=1,\dots,n.$ In addition, it follows that,  
$w(t_i) = O(|T| W_{u_{i\, 3}} + |T| w(u_{i\, 3}))= O(|T| w(u_{i\, 3}))$
for each $i = 1, \dots, m,$ whence also 
$w(\tilde{X}_i) =  O(|T| w(u_{i\, 3}))$.

Fix $i > 3$ and assume that the claim is true for each $i' < i.$

Let us first consider that case  $\pi_i = X_j,$ for some $1 \leq j \leq m.$ In this case 
$w(\pi_i) = w(\tilde{X}_j) =  O(|T| w(u_{j\, 3})).$ Since $u_{j \, 3} \prec X_j$ it follows, by induction hypothesis,  that
$w(u_{j\, 3}) \leq w(\pi_{i-1}) = O(|T|^{3(i-1)}),$ from which the desired result follows.

Alternatively, if, $\pi_i = u_j,$ for some $1 \leq j \leq n,$ then either we  have 
$w(u_j) = |T|^3 w(u_{j-1})$ or we have 
$w(u_j) = \sum_{X \prec u_j} w(\tilde{X}) \leq |T| w(\pi_{i-1}).$  In either case the desired bound follows by induction 
hypothesis.
\end{proof}

}

\section{The proof of lemma \ref{lemma:optimal_structure}}
We need two inequalities regarding the weights.

\noindent
{\em Fact 1} 
For each $1 \leq i' < i \leq m$ it holds that     
\begin{equation} \label{eq:weights}
w(t_i) > w(a_{i1}) > w(t_{i'}) + w(u_{i'\, 1}) + w(u_{i'\, 2}) + w(u_{i'\, 3})
\end{equation} 
{\em Proof of the fact.}
The first inequality follows by definition. In order to prove the second inequality let us consider the difference
$$Diff = w(a_{i1}) -  \left(w(t_{i'}) + w(u_{i'\, 1}) + w(u_{i'\, 2}) + w(u_{i'\, 3})\right) .$$
By definition we have
$$Diff = \sum_{j=1}^3 \left ( W_{u_{i j}} + \gamma(i,j) w(u_{i j}) \right) -
    \sum_{j=1}^3 \left ( W_{u_{i' j}} + (\gamma(i',j) + 3/2) w(u_{i' j}) \right).$$
    
    \medskip

{\em Case 1.} $u_{i 3} = u_{i' 3}.$ 
Note that   $\gamma(i,3) \geq 5 + \gamma(i',3)$,
Since  $W_{u_{ij}},W_{u_{i'j}}  \geq 0$ and
$  0 < \gamma(i,j) , \gamma(i',j) \leq |T|$ we get that

$$Diff \geq 5 w(u_{i 3}) - 3W_{u_{i' 3}} -(2|T|+3)w(u_{i' 2})  .$$

Let $\kappa$ be such that  $u_{i 3}=u_\kappa$. It follows from the definition of
the  function $w()$ that 
$$w(u_{i3}) = w(u_\kappa) = 1+  6 \max \{ W_{u_{\kappa}},|T|^3 w(u_{\kappa-1}) \} 
>  3W_{u_{i' 3}} + (2|T|+3)w(u_{i' 2}). $$
Thus, $Diff>0$.

\remove{
{\em Case 1.} $u_{i 3} = u_{i' 3}.$ Then, since $X_{i'} \prec X_{i}$ we must have $u_{i 2} \not \prec u_{i' 2},$ whence
$w(u_{i' 2}) \leq  w(u_{i 2})$ and $W_{u_{i' 2}} \leq  W_{u_{i 2}}.$ Therefore, using the non-negativity  of  $W_{u_{i 1}}$ and  $w(u_{i1}),$  
 we have

$$Diff > \left( \gamma(i,3) - \gamma(i',3) -3/2 \right) w(u_{i 3}) + 
\left( \gamma(i,2) - \gamma(i',2) -3/2 \right) w(u_{i 2}) -  
W_{u_{i' 1}} - \left(\gamma(i',1) + 3/2\right) w(u_{i' 1}).$$

Finally, we observe that under the standing hypothesis we have $\gamma(i,3) \geq \gamma(i',3) + 5.$ Moreover, 
$W_{u_{i' 1}}  \leq w(u_{i'1}) \leq w(u_{i'2}) \leq w(u_{i2})$ and 
$\gamma(i',1),  \gamma(i',2) \leq |T|,$ and $\gamma(i, 2) > 0.$ Therefore we have
$$Diff > (7/2) w(u_{i 3}) - (|T| + 3/2) w(u_{i 2}) - w(u_{i 2}) - (|T| + 3/2) w(u_{i2}) = 
(7/2) w(u_{i 3}) - 2(|T| + 2)  w(u_{i 2}) > 0$$ where the last inequality follows 
because of $w(u_{i 3}) \geq |T|^3 w(u_{i 2}).$
}

\medskip

{\em Case 2.} $u_{i' 3 } \prec u_{i 3}.$ Then, it must also hold that $X_{i'} \prec u_{i 3}.$ Therefore we have 
$$w(a_i) \geq W_{u_{i 3}} \geq w(\tilde{X}_{i'}) >  w(t_{i'}) + w(u_{i'\, 1}) + w(u_{i'\, 2}) + w(u_{i'\, 3}).$$  

\medskip

\noindent {\em Fact 2} 
For each $1\leq i \leq m$ and $\kappa = 1,\dots, 4,$ it holds that     
\begin{equation} \label{eq:weights2}
w(a_{i\kappa}) \geq  3 ( w(u_{i3}) +  w(u_{i2}) +   w(u_{i1}))  + W_{u_{i3}}
\end{equation} 
It follows directly from the definition of $w(a_{i\kappa})$ and the  the fact that $\gamma(i,j) \geq 3$ ($j=1,2,3$).

\remove{ it follows that 
 $w(a_{i1}) \geq 3 (\sum_{\kappa = 1}^3 3 w(u_{ik}) +W_{u_{i3}}$.
The inequality follows from the fact that $w(u_{i\kappa}) > 0.$
}


\remove{
\begin{definition}
Given a binary tree $T$ and a node $\nu \in D,$ such that $\nu$ is the left child of its parent $p(\nu),$ and $p(\nu)$ is the left child of its own parent,
  an upper swap on $\nu$ is the operation that
transforms $D$ into a new tree  as follows: (i) $\nu$ becomes the left child of the parent of $p(\nu)$; (ii) the node  $p(\nu)$ becomes the left child of $\nu;$
(iii) the left subtree of $\nu$ becomes the left subtree of $p(\nu)$ (which is now $\nu$'s left child); (iv) the right subtree of $\nu$ and the right subtree of 
 $p(\nu)$ are not changed.
\end{definition}
}

\bigskip
\noindent
{\bf Proof of Lemma \ref{lemma:optimal_structure}.}
Let $D$ be an optimal search tree for $(T,w)$.  

Let $\ell$ be the deepest  node in the left path of $D$ such that 
$D - D_{\ell}$ is the realization of $\pi_{i+1} \dots \pi_{n+m}$ for some $i=0,\dots, n+m.$ In particular, we take $i=n+m$ if 
$\ell$ is the root of $D,$ i.e., no upper part of $D$ looks like a realization of suffix of $\Pi.$

By contradiction, assume that $D$ is not a realization of $\Pi,$ in particular  $i > 0.$
We shall prove that by modifying $D_{\ell}$ in such a way that its top part becomes a realization of $\pi_{i}$ we obtain a 
new search tree with cost smaller than the cost of $D.$ The desired result will follow by  contradiction.
We consider the following cases:

\remove{
For otherwise, we can obtain a new search tree of smaller average cost, violating the optimality of $D.$ To see why, let us first consider the case
when the node $q_{t_j}$ is the right child of $q_{r_j}.$ }

\noindent
{\em Case 1.} $\pi_i = X_j,$ for some $j=1,2,\dots, m.$ First we argue  that $\ell \in \{ q_{t_j}, q_{r_j}\}.$ 
 Let $q_{\nu}$ (for some $\nu \in T$) be the parent of $q_{r_j}.$ 
If $\nu \in T_j$ we swap $q_{t_j}$ with $q_{\nu}$ otherwise we   swap $q_{r_j}$ with $q_{\nu}.$\footnote{When swapping we imply that the two nodes are 
exchanging position and they are carrying along also their right subtrees. This is possible because $q_{r_j}$ is the left child of $q_{\nu}.$}
Let $D'$ be the new tree so obtained.

If $\nu$ is a leaf in $T,$ then we have $cost(D') \leq cost(D) - w(t_j) + w(\nu) < cost(D)$ since $t_j$ is the leaf of largest weight in $D_{\ell}.$
Otherwise, it must be that $\nu = r_{j'}$ for some $j' < j.$ In this case, by (\ref{eq:weights}), we have 
$cost(D') \leq cost(D) - w(t_j) + w(t_{j'}) + w(u_{j' \, 3}) + w(u_{j' \, 2}) + w(u_{j' \, 1}) < cost(D).$ 
In either case we obtain a tree of average weight smaller than $D$, violating the optimality of $D.$ 

Alternatively, if $q_{t_j}$ is not the right child of $q_{r_j},$ then we swap  $q_{t_j}$ with its parent. 
Note that $q_{t_j}$ must be the left child of its parent.  By proceeding as above, 
we can prove that the resulting tree has cost smaller than $D,$ again a violation to the optimality of $D.$
\commento{I added this missing case,--- whether it is clear enough.}
Therefore, it must be $\ell \in \{ q_{t_j}, q_{r_j}\}.$ We now split the analysis according to this two possible cases.

\medskip
\noindent
{\em Subcase 1.1.} $\ell = q_{r_j}.$ Then, because of the assumption on $D-D_{\ell}$ and the search property, it follows that the right subtree of 
$q_{r_j}$ contains the nodes $q_{t_j}, q_{s_{j\, 3}}, q_{s_{j\, 2}}, q_{s_{j\, 1}}.$ Also, it is not hard to see that they must appear in this order, for otherwise 
by reordering them we would decrease the average cost of $D,$ since $w(t_j)  > w(s_{j\, 3}) > w(s_{j\, 2}) > w(s_{j\, 1}).$ Therefore the right subtree of $\ell$
coincides with the right subtree of $D_j^A.$ 

Suppose now w.l.o.g. that for each $\kappa = 2, 3,4,$ it holds that 
$q_{a_{j\kappa-1}}$ is closer  to the root of $D$ than $q_{a_{j\kappa}}$ 
For the sake of contradiction, assume that  $q_{a_{j1}}$ is not a child of $q_{r_j}.$  
Let $q_{\nu}$ be the parent of $q_{a_{j1}}.$ Note that $q_{a_{j1}}$ can only be the left child of $q_{\nu}.$  
By swapping  $q_{a_{j1}}$ with $q_{\nu}$    the resulting tree has smaller expected cost than $D,$ 
again in contradiction with the assumed optimality of $D.$ 
 In fact, if $\nu$  is a leaf in $T$  then it follows from inequality (\ref{eq:weights2}) that  $w(a_{j1}) > w(u_{j3}) \geq w(\nu)$.
 Otherwise, if 
 $\nu = r_{j'}$ for some  $j' < j,$ and then, by (\ref{eq:weights}) 
 we have that $w(a_{j1})$ is greater than the weight of the right subtree  of $q_{\nu}.$ 
The same arguments show that $q_{a_{j\kappa}}$ is the left child of $q_{a_{j\kappa-1}},$
for each $\kappa = 2,3,4.$  
   
 We can conclude that 
in the left path of $D,$ the nodes following $\ell$ are exactly $q_{a_{j1}}, \dots,q_{a_{j4}}$. Let $\ell'$ be the
left child of $q_{a_{j4}}.$ We have showed that in this subcase
$D_{\ell} - D_{\ell'}$ coincides with $D_{j}^A.$

\medskip
\noindent
{\em Subcase 1.2.} $\ell = q_{t_j}.$ There is nothing to prove about the right 
subtree of $\ell.$ 
 In order to prove that in the left path of $D,$ the node $\ell$ is followed by 
 $q_{a_{j1}}, \dots,q_{a_{j4}}$\footnote{We are again assuming, w.l.o.g., that 
 for each $\kappa = 2, 3,4,$ it holds that 
$q_{a_{j\kappa-1}}$ is closer  to the root of $D$ than $q_{a_{j\kappa}}.$}
we proceed as before. Assume (by contradiction) that $q_{a_{j1}}$ is not a child
of $q_{r_j}.$ Let $q_{\nu}$ be the parent of  $q_{a_{j1}}.$ Note that $q_{a_{j1}}$ can only be the left child of $q_{\nu}.$ 
We swap $q_{a_{j1}}$ with  $q_{\nu}.$ Let $D'$ be the resulting search tree. If $\nu$ is $r_j$ or a leaf in $T_j \setminus \{t_j\},$
we have that $cost(D') = cost(D) - w(a_{j1}) + w(X_j) < 0,$ where $w(X_j)$ accounts for the weight of the right subtree of
$q_{\nu}$ and the last inequality follows by  (\ref{eq:weights2}).
On the other hand, if $\nu$ is either a leaf in $T$ or is equal to $r_{j'}$ for some $j' < j,$ then we can apply the same argument as in 
Subcase 1.1, to reach the same conclusion, i.e., we violate the optimality of $D.$ 

Therefore, we conclude that  $q_{a_{j1}}$ is the left child of $q_{\ell}.$ Repeating the same argument we can also show that
$q_{a_{j\kappa}}$ is the left child of $q_{a_{j\kappa-1}},$ for each $\kappa = 2,3,4.$  
 Let $\ell'$ be the
left child of $q_{a_{j4}}.$ We have showed that in this subcase, 
$D_{\ell} - D_{\ell'}$ coincides with $D_{j}^B.$ 

We can conclude that in both subcases of Case 1, the tree $D - D_{\ell'}$  is realization of $\pi_i, \dots, \pi_{n+m}$ against the assumption  
that $\ell$ is the deepest node for which such a condition holds.
 
\medskip

\noindent
{\em Case 2.} $\pi_i = u_j,$ for some $j=1,2,\dots, n.$


Let us consider the set of leaves $L$ of  $T^b$ which  are associated with $u_j$ and such that they are not queried 
 in $D - D_{\ell}.$ Since  $D - D_{\ell}$ is a realization of $\pi_{i+1} \dots \pi_{n+m},$ the leaves of $T^b$ which are not in $L$ and
 are queried in $D_{\ell}$ are either in  $\bigcup_{X \prec u_j} \tilde{X}$  or are associated to $u_{j'}$ for some $j' < j$.
For the sake of contradiction we assume that one of the  first $|L|$ nodes in the left path 
of $D_{\ell}$ does not correspond to a leaf in $L$.

Let us construct a tree $D'$ from $D_{\ell}$  as follows:
first we construct an auxiliary tree by removing from $D_{\ell}$ all the nodes corresponding to the leaves in $L$.
 Then, we add a left path with these nodes to the top of this auxiliary tree.
Our assumption that one of the  first $|L|$ nodes in the left path 
of $D_{\ell}$ does not correspond to a leaf in $L$ implies that 

$$cost(D') \leq cost(D_{\ell}) -w(u_j) + |L| \sum_{X \prec u_j} w(\tilde{X}) +  3\cdot |L| \cdot  \sum_{u \prec u_j} w(u)  
$$.

The negative term in the equation above is because  the sum of the levels of  
the nodes associated with $u_j$ in $D_{\ell}$ is at least 7
while  this sum  is exactly 6 in $D'$. 
The other terms are due to the fact that the level
of a node can increase  by at most $|L|$ units 
in our construction. 
The definitions of $W_{u_j}$ and $w(u_j)$ imply that

$$ cost(D') \leq cost(D_{\ell}) -w(u_j) + |L| W_{u_j} + |L| \cdot |T| \cdot  w(u_{j-1}) $$

Since $|L|\leq 3$ and  $w(u_j) > 6 \max \{ W_{u_j},|T|^3 w(u_{j-1}) \}$  we get that
$ cost(D') < cost(D_{\ell})$. This implies, however, that $D$ can be improved, a contradiction.

Thus, the   $D_{\ell}$'s $|L|$ top levels coincide with a 
sequential search tree for $L.$ Let $\ell'$ the left most query of such sequential  search. 
Therefore, $D-D_{\ell'}$ is  realization of $\pi_i \dots, \pi_{n+1},$ which contradicts also in this Case 2 the hypothesis that  $\ell$ is the
deepest node for which such a condition holds.

The proof is complete. \hfill{\qed}

\section{The proof of Lemma \ref{lemma:key_2}}

\noindent
{\bf Lemma 2.} 
{\em Let $D^*$ be an optimal binary search tree for $(T, w).$
Let ${\cal Y} \subseteq {\cal X}$ be such that $D^*$ is a realization of $\Pi$ w.r.t. $\cal Y.$
We have that $cost(D^*) \leq  cost(D^A) - \frac{1}{2}\sum_{u \in  U} w(u)$  if and only if  ${\cal Y}$ is a solution for the 
X3C instance $\mathbb{I} = (U, {\cal X}).$}

\medskip

\begin{proof}
We start proving the {\em only if} part. Assume that $cost(D^A) - cost(D^*) \geq \frac{1}{2}\sum_{u \in  U} w(u).$
We shall use induction on $j$ to prove that for each $j=n, \dots, 1$ there exists exactly one $X \in {\cal Y},$ such that 
$u_j \in X.$

Fix $j^* \leq n$ and assume that for every $j>j^*$ it holds that there exists exactly one $X \in {\cal Y}$ such that 
$u_{j} \in X.$

Suppose  that  there is no  $i \in \{1,\dots, m\}$ such that 
$u_j^* \in X_i \in {\cal Y}.$
We can rewrite (\ref{eq:Delta_cost}) as follows:
$$ cost(D^A)-cost(D^*) = \sum_{j=1}^n \mathop{\sum_{X_i \in {\cal Y}}}_{u_j \in X_i} \left( \frac{w(u_j)}{2} + 
\Gamma(i,j) w(u_j) \right),$$
where $\Gamma(i,j) = \gamma(i,\kappa) - d^A_B(q_{s_{i \, \kappa}}),$ and $\kappa  \in \{1,2,3\}$ such that $s_{i\, \kappa} = u_j.$

Now, since we are assuming that for all $j > j^*$ there exists only one $i$ such that $X_i \in {\cal Y}$ and $u_{j} \in X_i,$
by the definition of $d^A_B(\cdot)$ and $\gamma(i,\kappa),$ we have  $\Gamma(i,j) = 0.$ So we obtain
$$ cost(D^A)-cost(D^*) = \sum_{j > j^*} \frac{w(u_{j})}{2} + \sum_{j < j^*} \mathop{\sum_{X_i \in {\cal Y}}}_{u_{j}\in X_i} \left( \frac{w(u_{j})}{2} + 
\Gamma(i,j) w(u_{j}) \right),$$
where we also used the assumption that no $X \in {\cal Y}$ contains $u_{j^*}$ and therefore 
$u_{j^*}$ does not contribute to the sum.
   
Now we can observe that, for each $j < j^*,$ there are at most $3$ set in ${\cal X}$ containing $u_{j}.$ 
Moreover, $\Gamma(i, j)$ being a difference of levels in $D^*$ can be bounded by $|T|.$ Also 
$w(u_{j}) \leq w(u_{j^*})/6|T|^3,$ for each $j < j^*$. Therefore, we have the desired contradiction:
$$ cost(D^A)-cost(D^*) \leq \sum_{j > j^*} \frac{w(u_{j})}{2} + 3(j^*-1)(|T| + 1/2) \frac{w(u_{j^*})}{6|T|^3} < \sum_{j  > j^*} \frac{w(u_{j})}{2} 
+ \frac{w(u_{j^*})}{2}\leq 
\frac{1}{2}\sum_{u \in  U} w(u).$$

\smallskip

Suppose now that  there are  $\kappa > 1$ subsets  in $ {\cal Y}$ 
that  contain  $u_{j^*}.$  Rewriting (\ref{eq:Delta_cost}) as before, we obtain:

$$ cost(D_A)-cost(D^*) \leq \sum_{j > j^*} \frac{w(j)}{2}  +
\mathop{\sum_{X_i \in {\cal Y}}}_{u_{j^*}\in X_i} \left( \frac{w(u_{j^*})}{2} + 
\Gamma(i,j^*) w(u_{j^*}) \right) + 
\sum_{j < j^*} \mathop{\sum_{X_i \in {\cal Y}}}_{u_{j }\in X_i} \left( \frac{w(u_{j})}{2} + 
\Gamma(i,j) w(u_{j})\right) .$$

Let us observe that among the $\kappa$ sets $X_i \in {\cal Y}$ such that $u_{j^*} \in X_i$ only one contributes with a
positive weight $w(u_{j^*})/2$ since $\Gamma(i,j^*)=0.$ 
For the others, we have a negative contribution of at least $w(u_{j^*})/2,$ since $\Gamma(i,j^*)$ becomes negative.
Moreover, for the $j < j^*$ we can repeat the argument we used in the previous case. Therefore we obtain the desired contradiction:

$$ cost(D_A)-cost(D^*) \leq \sum_{j>j^*} \frac{w(j)}{2}  
-\frac{\kappa-2}{2} w(u_{j^*}) + 3({j^*}-1)(|T| + 1/2) \frac{w(u_{j^*})}{6|T|^3} < 
\sum_{j  \geq j^*} \frac{w(u_{j})}{2} < 
\frac{1}{2}\sum_{u \in  U} w(u).$$
This concludes the inductive argument and the  proof of the {\em only if} part.

\medskip

In order to prove the   {\em if} part of the statement we notice that if ${\cal Y}$ is a solution for $\mathbb{I}$ then
for each $j=1,\dots, n$ there exists exactly one index $i$ such that  $X_i \in {\cal Y}$ and $u_j \in X_i.$ 
Then,  the desired result follows directly by equation  (\ref{eq:Delta_cost}), and by the fact that in this case the  
 definition of $d^A_B(\cdot)$ and  $\gamma(\cdot,\cdot),$ yields $\Gamma(i,j) = 0.$ 
\end{proof}


\section{The proof of Lemma \ref{lemma:bounded_optimal_structure}}

{\bf Proof of Lemma \ref{lemma:bounded_optimal_structure}.}
Let $D$ be an optimal search tree for $(T^b,w)$. 

Let $\ell$ be the deepest  node in the left path of $D$ such that 
$D - D_{\ell}$ is the realization of $\pi_{i+1} \dots \pi_{n+m}$ for some $i=0,\dots, n+m.$ In particular, we take $i=n+m$ if 
$\ell$ is the root of $D,$ i.e., no upper part of $D$ looks like a realization of some suffix of $\Pi.$

By contradiction, assume that $D$ is not a realization of $\Pi$, whence $i > 0.$
We shall prove that by modifying $D_{\ell}$ in such a way that its top part becomes a realization of $\pi_{i}$ we obtain a 
new search tree with cost smaller than the cost of $D.$ The desired result will follows by  contradiction.
We consider the following cases:

\noindent
{\em Case 1.} $\pi_i = X_j,$ for some $j=1,2,\dots, m.$

In this case, our assumption regarding  
$\ell$ implies that if a node $\nu  \in  D_{\ell}$
is associated with a leaf $\ell'$ in $T^b$ then 
$\ell'$  either corresponds to an element
$u \in U$ such that 
$ u \prec X_j$ or $\ell' \in \tilde{X}_{j'}$ such that
$X_{j'} \preceq X_j$.
 Let $\kappa$ be such that $X_j \in H_{\kappa}.$ 
We need to prove the following claim

\newtheorem{claim}{Claim}

\noindent 
\begin{claim}   $\ell \in \{ q_{t_j}, q_{r_j}\}.$ 
\end{claim}
\begin{proof}
We shall show it by contradiction. We split the proof into cases I  and II.

\noindent
\medskip
{\em Case I.} Suppose that the node $q_{t_j}$ is the right child of $q_{r_j}.$    Let $q_{\nu}$ (for some $\nu \in T$) be the parent of $q_{r_j}.$ 
We have two cases according as $q_{r_j}$ is a right or a left child of $q_{\nu}.$

\medskip

{\em Subcase I.a}~ $q_{r_j}$ is a right child of $q_{\nu}.$ 
 Note that because of the search tree property 
$\nu$ must be an ancestor of $h_{\kappa}$ in $T^b.$ 
We perform a left rotation 
 on $q_{\nu}.$  
  Let $D'$ be the new tree obtained. We have that 
$ cost(D') \leq cost(D) - w(t_j) +w(\alpha)$, where $\alpha$ is the left subtree of $q_{\nu}.$
We observe if a node in $\alpha$ corresponds to a  leaf $\ell'$ then 
$\ell'$ must be in $T^b \setminus H_{\kappa}$.

Thus, the nodes of $\alpha$ can take care of:

(a) leaves that are associated to some $u \in U,$ such that $u \prec u_{j 3}.$
The sum of the weights of these leaves is at most $|T| \cdot w(u_{j 3})/6|T|^3 < w(u_{j 3})/2$;

(b) at most two leaves  associated with 
$u \in U$ such that $u=u_{j3}$. The fact that every $u \in U$ appears
in at most three sets of ${\cal X}$ together with the fact that $s_{j3} \in H_{\kappa}$ 
explain that we have at most two leaves; 

(c)  leaves in $\tilde{X}_{j'}$  such that
$X_{j'} \prec u_{j 3}$. The sum of the weights of these leaves sum at most $W_{u_{j3}}$.

\smallskip

Thus, we can conclude that $ w(\alpha) \leq 2.5 w(u_{j 3})+W_{u_{j3}}$.
Since $w(t_j)> 2.5 w(u_{j 3})+W_{u_{j3}}$ we
conclude that $  cost(D') < cost(D),$ contradicting the optimality of $D.$

\remove{ 
In particular, by induction hypothesis, the leaves of $T^b$ that  $\alpha$ takes care of, are either from the sets $X_{j'},$ with $j' < j$ or 
are leaves associated to some $u \in U,$ such that $u \prec u_{j 3}.$ The number of leaves associated to $u \prec u_{j 3}$ 
is clearly upper bounded by $3 |T|,$ hence their total 
weight is at most $3 |T| w(u_{j 3})/|T|^3 < 1/2 w(u_{j 3}).$ 
Also notice that, because of the way $H_{\kappa}$ is defined, we have that
for each $j' < j,$ it holds $X_{j'} \prec u_{j 3}.$  
Therefore, we have $w(t_j) > W_{u_{j 3}} + \frac{1}{2} w(u_{j 3}) = \sum_{X \prec u_{j 3}} w(\tilde{X}) + \frac{1}{2} w_{u_{j 3}} \geq 
w(\alpha).$  Therefore, $cost(D) - cost(D') > 0,$ contradicting the optimality of $D.$
}

\medskip
{\em Subcase I.b}~ $q_{r_j}$ is a left child of $q_{\nu}.$ 
This implies that $\nu$ is not an ancestor of $r_j$ in $T^b$.
Let $D'$ be a  tree obtained as follows: 
we swap  $q_{r_j}$ with $q_{\nu}$ if $\nu$ is not in $T_j$; otherwise,
we swap $q_{t_j}$ with $q_{\nu}$. 
Let $\alpha$ be the right subtree of $q_{\nu}.$ 
Again, we have 
$ cost(D') \leq cost(D) - w(t_j) + w(\alpha)$.

If $\nu \notin H_{\kappa}$ then the analysis is identical to the one employed in Subcase I.a
because $\alpha$ can take care of the same leaves considered in that case.

If   $\nu$  is a leaf in  $H_{\kappa}$ then
$w(t_j) > w(\nu) = w(\alpha)$ because $t_j$ is the heaviest leaf 
among the leaves in $H_{\kappa}$ that corresponds to a node in $D_{\ell}$.
Finally, if $\nu$ is an internal node in $ H_{\kappa}  \setminus \{h_{\kappa}\}$
then $\nu =r_{j'}$ for some $j'<j$ and it follows from inequality (\ref{eq:weights})
that $w(t_j) > w(t_j') + w(u_{j'3})+w(u_{j'2})+w(u_{j'1}) = w(\alpha)$.

\medskip
In either Subcase we obtain a tree of cost smaller than $D$ violating the optimality of $D.$ 

\medskip
\noindent
{\em Case II.}
Alternatively, if $q_{t_j}$ is not the right child of $q_{r_j},$ then we can proceed as before.
We consider the case where $q_{t_j}$ is the right child of its parent and also
the case where it is the left child. In the former case we apply a left rotation and
in the latter a simple swap. 
Again  we can prove that the resulting tree has cost smaller than $D,$ a violation to the optimality of $D.$

\smallskip
The proof of the claim is complete.
\end{proof}
Therefore, it must be $\ell \in \{ q_{t_j}, q_{r_j}\}.$ We now split the analysis according to this two cases.

\noindent
{\em Subcase 1.1.} $\ell = q_{r_j}.$ Then, because of the assumption on $D-D_{\ell}$ and the search property, it follows that the right subtree of 
$q_{r_j}$ contains the nodes $q_{t_j}, q_{s_{j\, 3}}, q_{s_{j\, 2}}, q_{s_{j\, 1}}.$ Also, it is not hard to see that they must appear in this order, for otherwise, 
by reordering them we would decrease the average cost of $D,$ since $w(t_j)  > w(s_{j\, 3}) > w(s_{j\, 2}) > w(s_{j\, 1}).$ Therefore the right subtree of $\ell$
coincides with the right subtree of $D_j^A.$ 

Let us assume w.l.o.g that the level of $q_{a_{jk}}$ is smaller than or equal to the level
$q_{a_{jk'}}$ in $D$,  for $k < k'$.
First, we argue that the left child  of $\ell$ must be $q_{a_{j1}}$.
Assume that  $q_{a_{j1}}$ is not the left child of $\ell$
and let $\nu$ be the parent of $q_{a_{j1}}$. We have two cases:

\medskip
{\em A}. $q_{a_{j1}}$ is a right child of $\nu$. 

We perform a left rotation  on $q_{\nu}.$  Let $D'$ be the new tree obtained. We have that 
$ cost(D') \leq  cost(D) - w(q_{a_{j1}}) +w(\alpha)$ where $\alpha$ is the left subtree of $\nu.$
Note that the search property assures that  $\nu$ is an ancestor of $h_{\kappa}$.
Thus, the analysis of Subcase I.a in the above Claim 1, shows that the the sum of the weights
of the leaves that $\alpha$ can takes care is upper bounded by $2.5 w(u_{j3}) + W_{u_{j3}}$.
Since $w(q_{a_{j1}}) > 3w(u_{j3}) + W_{u_{j3}}$ we conclude that
$cost(D') < cost(D)$.

\medskip
{\em B}.
$q_{a_{j1}}$ is a left child of $\nu$. 
In this case, we swap $q_{a_{j1}}$ and $\nu$.
Let $D'$ be the new tree obtained. We have that 
$ cost(D') \leq cost(D) - w(q_{a_{j1}}) +w(\alpha)$ where $\alpha$ is the right subtree of $\nu.$
Note that $\nu$ is not an ancestor of $a_{j1}$ in $T^b$.

If $\nu \notin H_{\kappa}$ the arguments 
employed in subcase I.A shows that 
$w(\alpha) \leq 2.5 w(u_{j3}) + W_{u_{j3}}$.
Since $w(q_{a_{j1}}) > 3w(u_{j3}) + W_{u_{j3}}$ we conclude that
$cost(D') < cost(D)$.

If $\nu \in T_{j'}$, with $T_{j'} \in H_{\kappa}$ and $j' < j$,
it follows from inequality (\ref{eq:weights}) that $w(a_{j1}) > w(\alpha) $.
If $\nu \in T_{j}$ it follows from inequality (\ref{eq:weights2}) that $w(a_{j1}) > w(\alpha) $.
Finally, if $\nu = a_{j'k}$ with $j'<j$ we have that $w(a_{j1}) > w(a_{j'k}) =w(\alpha) $

\medskip

We can conclude that   $q_{a_{j1}}$is the left child of $\ell$. 
Since $w(a_{j1})=w(a_{j2})=w(a_{j3})=w(a_{j4})$, the same arguments show that the nodes 
following $a_{j1}$ 
in the left path are $q_{a_{j2}}$, $q_{a_{j3}}$ and $q_{a_{j4}}$.
Let $\ell'$ be the
left child of $q_{a_{j4}}.$ We have showed that in this subcase $D_{\ell} - D_{\ell'}$ coincides with $D_{j}^A.$

\medskip

\noindent
{\em Subcase 1.2.} $\ell = q_{t_j}.$ There is nothing to prove about the right subtree of $\ell.$ 
On the other hand, in order to  prove that the nodes
following $\ell$ 
in the left path of $D$ are exactly $q_{a_{j1}}, q_{a_{j2}}$, $q_{a_{j3}}$ and $q_{a_{j4}},$ we can proceed as in Subcase 1.1.
 The only additional case to be taken care of, in the argument by contradiction used there, 
 is  when the parent of $q_{a_{j1}}$ is $q_{r_j}.$ However, in this case we can employ the same argument we used for the analogous situation 
in Subcase 1.2. of the proof of Lemma  \ref{lemma:optimal_structure}.
 Let $\ell'$ be the
left child of $q_{a_{j4}}.$ We have showed that in this subcase, 
$D_{\ell} - D_{\ell'}$ coincides with $D_{j}^B.$ 

We can conclude that in both Subcase  1.1 and 1.2, the tree $D - D_{\ell'}$  is a realization of $\pi_i, \dots, \pi_{n+m}$ against the assumption  
that $\ell$ is the deepest node for which such a condition holds.
 
\medskip

\noindent
{\em Case 2.} $\pi_i = u_j,$ for some $j=1,2,\dots, n.$ 

The proof is identical to that employed for  Case 2 of Lemma 
\ref{lemma:optimal_structure}
 \hfill{\qed}

\section{The proof of Lemma \ref{lemma:1_ESA}}
Let $x$ and $x^*$ be the nodes queried at the root of $D^\prime$ and $D^*$, respectively. W. l. o. g. we assume $x \neq x^*$, as otherwise the lemma trivially holds. We can also assume that $x^*$ is a node from $T_x$, because the opposite case is analyzed analogously.

\emph{Case 1:} $w(T_x) \leq w(T - T_x)$. In other words, $w(T_x) \leq w(T)/2$. As any path from $r(D^*)$ to a leaf in $D^*$ contains $r(D^*)$ and $T - T_x$ does not contain $x^*$, Lemma~\ref{subSS} states that the depth of any leaf in $D^*_1$ is at least by one smaller than it is in $D^*$. The lemma also implies that the depth of any leaf in $D^*_0$ is not greater than it is in $D^*$. So we have
\begin{align*}
cost(D^\prime) = & \ w(T) + cost(D_0^*) + cost(D_1^*) 
\\
	 \leq & \ w(T) + \sum_{v \in T_x} w(v) d(r(D^*),l_v) +  \sum_{v \in T - T_x} w(v) \bigl( d(r(D^*),l_v) - 1 \bigr)
	\\
	= & \ w(T) + cost(D^*) - w(T - T_x) \leq cost(D^*) + w(T)/{2} \ .
\end{align*}

\emph{Case 2:} $w(T_x) > w(T - T_x)$. Let $x_1,\ldots,x_n$ be the nodes successively queried when the path $r(D^*) \leadsto r(D^\prime)$ is traversed in $D^*$. In particular, $x_1 = x^*$ and $x_n = x$. Let $k < n$ be such that $x_i$ is a node from $T_x - \{x\}$ for $i = 1,\ldots,k$ and $x_{k+1} \notin T_x - \{x\}$. 

In this extended abstract we assume that $w(T_{x} - T_{x_i}) > 0$ for $i = 1,\ldots,k$. The case of $w(T_{x} - T_{x_i}) = 0$ can only occur when there is tie regarding the choice of node $x$ in step (1) of the algorithm, and then the above scenario can be avoided by employing a suitable tie breaking rule. In the full paper we will show by a more intricate case analysis that the approximation factor holds regardless of the tie breaking rule.

For $i = 1,\ldots,k$ we know that $w(T_{x_i}) < w(T - T_{x_i})$, because otherwise, using the assumption that  $w(T_{x} - T_{x_i}) > 0,$ we would have
$w(T_{x_i}) - w(T - T_{x_i}) = w(T_{x_i}) - w(T_x - T_{x_i}) - w(T - T_x) = w(T_x) - w(T- T_x) - 2 w(T_x - T_{x_i}) < w(T_x) - w(T-T_x),$ and so
$x_i$ would have been chosen instead of $x$ in step~(1) of the algorithm.
\remove{
$w(T_x) - w(T - T_x) = w(T_{x^i}) + w(T_x - T_{x^i}) - w(T - T_x) > w(T_{x^i}) - w(T_x - T_{x^i}) - w(T - T_x) = w(T_{x^i}) - w(T - T_{x^i}) \geq 0$, so $x^i$ would have been chosen instead of $x$ in step (1) of the algorithm.
}

From this fact, it follows that $w(T_{x_i}) \leq w(T - T_x)$ for $i = 1,\ldots,k$. This is because otherwise $w(T_x) - w(T - T_x) = w(T_{x_i}) + w(T_x - T_{x_i}) - w(T - T_x) >  w(T - T_x) +  w(T_x - T_{x_i}) - w(T_{x_i}) = w(T - T_{x_i}) - w(T_{x_i}) \geq 0$, so $x_i$ would have been chosen instead of $x$ in step (1).

Let $T^\prime := \bigcup_{i = 1}^k T_{x_i}$ and let $T^{\prime\prime} := T_x - T^\prime$. Note that $T^\prime$ is a forest in general and $T^\prime \cup T^{\prime\prime} = T_x$. We are going to reason about the search tree depths of the nodes in $T - T_x$, $T^\prime$, and $T^{\prime\prime}$ separately.

$D^*_0$ queries all nodes from $T^\prime$, and Lemma~\ref{subSS} states that the depth of those nodes is not greater in $D^*_0$ than it is in $D^*$. 

The nodes from $T^{\prime\prime}$ are as well all queried in $D^*_0$. For these nodes we know that in $D^*$ the node $x_{k+1}$ is queried before them. As $x_{k+1}$ is not queried by $D^*_0$, 
the depth of each node from $T^{\prime\prime}$ in $D^*_0$ is by at least by one smaller than it is in $D^*$.

Finally, the leaves in $D^*$ corresponding to the nodes from $T - T_x$ are descendants of the nodes in $D^*$ querying $x_1,\ldots,x_k$. 
These $k$ nodes are not contained in $D^*_1$, so the depth of each leaf in $D^*_1$ is at least by $k$ smaller than it is in $D^*$. Combining the findings, we obtain
\begin{align*}
cost(D^\prime) = & \ w(T) + \sum_{v \in T^\prime} w(v) d(r(D_0^*),l_v) + \sum_{v \in T^{\prime\prime}} w(v) d(r(D_0^*),l_v)
+ \sum_{v \in T - T_x} w(v) d(r(D^*_1),l_v)
\\ 
\leq & \ w(T) + \sum_{v \in T^\prime} w(v) d(r(D^*),l_v) + \sum_{v \in T^{\prime\prime}} w(v) \bigl(d(r(D^*),l_v) - 1 \bigr)
+ \sum_{v \in T - T_x} w(v) \bigl(d(r(D^*),l_v) - k\bigr)
\\
= & \ w(T) + cost(D^*) - w(T^{\prime\prime}) - k w(T - T_x) \ .
\end{align*}
As $T^\prime = T - ((T - T_x) \cup T^{\prime\prime})$, we have $w(T^\prime) = w(T) - w(T - T_x) - w(T^{\prime\prime})$, so
\[
cost(D^\prime) \leq cost(D^*) + w(T^\prime) - (k-1) w(T - T_x) \ .
\]
We have argued above that $w(T_{x_i}) \leq w(T - T_x)$ for $i = 1,\ldots,k$. Therefore, $w(T^\prime) = w(\bigcup_{i=1}^k T_{x_i}) \leq \sum_{i=1}^k w(T_{x_i}) \leq k w(T - T_x)$, and
\[
	cost(D^\prime) \leq cost(D^*) + k w(T - T_x) - (k-1) w(T - T_x) = cost(D^*) + w(T - T_x) \leq cost(D^*) + w(T)/2 \ .
\]
\qed


	\section{An FPTAS for Searching in Bounded-Degree Trees}

	\subsection{Algorithm for $\mathcal{P}^B(F,P)$}
	
	In this section we complete the correctness proof of the proposed algorithm for solving $\mathcal{P}^B(F,P)$. It has already been argued in Section \ref{sec:search} that the algorithm always returns a feasible solution. In addition, in Case 1 of the algorithm, the returned solution is also optimal. Here we prove the optimality for the second case:

		\paragraph{Case 2:} {\em $F$ is a tree $T_v$.} Let $D^*$ be an optimal solution for $\mathcal{P}^B(T_v, P)$. Consider the internal node of $D^*$ assigned to $v$; since $D^*$ is compatible with $P$ and since this node belongs to the left path of $D^*$, it corresponds to a node $p_i$ of $P$. Thus, we denote this internal node of $D^*$ assigned to $v$ by $\bar{p}'_i$. Let $\bar{z}'$ be the leaf of $D^*$ assigned to $v$ and notice that $\bar{z}'$ lies in the left path of the right subtree of $\bar{p}'_i$. We construct $\bar{D}'$ from $D^*$ by essentially applying the inverse of Step \ref{step3} of the algorithm: remove from $D^*$ the right subtree of $\bar{p}'_i$; this removed subtree becomes the subtree of $\bar{p}'_i$; assign $\bar{p}'_i$ as blocked and remove $\bar{z}'$. (One can use Figures \ref{fig:case1}.d and \ref{fig:case1}.c to better visualize this construction.)

	The tree $\bar{D}'$ is actually an EST for the forest $\{T_{c_1(v)}, \ldots, T_{c_{\delta(v)}(v)}\}$ and has height at most $B$. Now construct $\bar{P}'$ by taking the left path of $\bar{D}'$, setting all the non-blocked nodes as unassigned and also setting every node after $\bar{p}'_i$ as unassigned. Clearly $\bar{D}'$ is compatible with $\bar{P}'$ and thus feasible for $\mathcal{P}^B(\{T_{c_1(v)}, \ldots, T_{c_{\delta(v)}(v)}\}, \bar{P}')$. 
	
	Notice, however, that $\bar{P}'$ starts with the prefix of $P$ until $p_i$ (in terms of its assignment), then it has a blocked node corresponding to $p_i$ and then some unassigned nodes. Let $\bar{t}$ be the number of nodes in $\bar{P'}$.
Since the last node of $\bar{P}'$ comes from the parent of $\bar{z}'$ in $D^*$ and $D^*$ has height at most $B$, 
we have that $ \bar{t} \le B$. Thus,  the path $\bar{P}'$ coincides with the path $P^{i,t}$ constructed by the 
algorithm  when $t=\bar{t}$. 

	It is easy to see that the tree $D^{i, \bar{t}}$, as defined in the algorithm, has cost  
	$$	OPT(\mathcal{P}^B(\{T_{c_1(v)}, \ldots, T_{c_{\delta(v)}(v)}\}, P^{i,\bar{t}}) + \bar{t} \cdot w(v)=
	 OPT(\mathcal{P}^B(\{T_{c_1(v)}, \ldots, T_{c_{\delta(v)}(v)}\}, \bar{P}') + \bar{t} \cdot w(v),$$ 
		which is at most $cost(\bar{D}') + \bar{t} w(v)$ due to the feasibility of $\bar{D}'$. Finally, notice that this last quantity is actually the cost	 of $D^*$, so $cost(D^{i, \bar{t}}) \le cost(D^*)$. Since the procedure returns a solution which is at least as good as $D^{i, \bar{t}}$, its optimality follows.

		\subsection{Proof of Lemma \ref{lemma:geoDec}}

				By means of contradiction suppose $v^* \in D^*$ with $\dist{r(D^*)}{v^*}{D^*} \ge c$ but $w(D^*_{v^*}) > \alpha \cdot w(D^*)$. Let $\tilde{T}$ be the subtree of $T$ associated with $v^*$ and let $x$ be the root of $\tilde{T}$.

Let $y$ be a node in $T_x$ to be specified later. Let $T^0 = T - T_y$ and $T^i = T_{c_i(y)},$ for $i=1, \dots, \delta(y).$ Moreover, 
let $D^i$ be the search tree for $T^i$  obtained from $D^*$ via Lemma \ref{subSS}.  
We shall construct a new search tree $D'$ for $T$ as follows:  
the root of $D'$ is assigned to $y$; the left tree of $r(D')$ is the search tree $D^0$; in the right tree of $r(D')$ we build a left path containing nodes corresponding to queries for $c_1(y), c_2(y), \ldots, c_{\delta(y)}(y)$ and we  make
$D^i$ becomes the right subtree of  node querying $c_i(y).$

\remove{		Now we select a node $y$ in $T_x$ (to be specified later) and construct a new search tree $D'$ for $T$ as follows: the root of $D'$ is assigned to $y$; the left tree of $r(D')$ is a search tree $D^0$ for $T - T_y$ obtained from $D^*$ via Lemma \ref{subSS}; in the right tree of $r(D')$ we build a left path containing nodes corresponding to queries for $c_1(y), c_2(y), \ldots, c_{\delta(y)}(y)$ and append to the right of these nodes the search trees $D^i$ for $T_{c_i(y)}$ also obtained from $D^*$ via Lemma \ref{subSS} (Figure \ref{fig:img}). In order to simplify the notation let $T^0 = T - T_y$ and $T^i = T_{c_i(y)}$; notice that for every $i$, $D^i$ is a search tree for $T^i$. }
	
		It is easy to see that the cost of $D'$ is at most $\sum_{i = 0}^{\delta(y)} cost(D^i) + (\Delta(T) + 1) \cdot w(T)$. We claim  that,  for a suitable choice of $y$, $D'$ improves over $D^*$. 
For this, let  $S$ be the set of nodes of $T_x$ which are queried in the path $r(D^*) \leadsto v^*$. We distinguish the following cases. 	

\medskip \noindent
{\bf Case 1:} $|S| \ge \frac{2 c}{3}$. Set $y$ as a node in $T_x$ such that $|T_y \cap S| \ge \frac{|S|}{2}$ and $|T_{c_i(y)} \cap S| \le \frac{|S|}{2}$ for every child $c_i(y)$ of $y$ and construct $D'$ as described previously. To find such a node $y,$ traverse  $T_x$ starting at its root and proceeding as follows: if $u$ is the current node then move to the child $v$ of $u$ with largest $|T_v \cap S|$; the traversal ends when $|T_u \cap S| \le \frac{|S|}{2}$. The parent of the node where the traversal ends is the desired $y.$
	
	To bound the cost of $D'$ we first consider the cost of a particular tree $D^i$. From its construction we have that $\dist{r(D^i)}{l_u}{D^i} \le \dist{r(D^*)}{l_u}{D^*}$ for any node $u \in T^i$. Moreover, for any node $u \in T^i \cap \tilde{T}$ the path $r(D^*) \leadsto l_u$ contains $v^*$ and therefore it contains $|S \setminus T^i|$ queries to nodes in $T_x \setminus T^i$. Since these nodes were removed in the construction of $D^i$, we have that for every $u \in T^i \cap \tilde{T}$

	%
	\begin{equation*}
		\dist{r(D^i)}{l_u}{D^i} \le \dist{r(D^*)}{l_u}{D^*} - |S \setminus T^i| \le \dist{r(D^*)}{l_u}{D^*} - \frac{|S|}{2} \,,
	\end{equation*}
	where the last inequality follows from the definition of $y$. It follows that 
$$cost(D^i) \le \sum_{u \in T_i}  \dist{r(D^*)}{l_u}{D^*} \cdot w(u) - \frac{|S| \cdot w(T^i \cap \tilde{T})}{2}.$$
Combining  this bound with our upper bound on the cost of $D'$ we get that 
\begin{equation*}
	 cost(D') \le  cost(D^*) -  \dist{r(D^*)}{l_y}{D^*} \cdot w(y) - \frac{|S| w(\tilde{T} - y)}{2} + (\Delta(T) + 1) \cdot w(T) \,.
	\end{equation*}
	We claim that actually $cost(D') \le cost(D^*) - \frac{|S| w(\tilde{T})}{2} + (\Delta(T) + 1) \cdot w(T)$. To see this, first suppose $y \in \tilde{T}$; then $\dist{r(D^*)}{l_y}{D^*} \cdot w(y) \ge |S| \cdot w(y)$ and the claim holds. In the other case where $y \notin \tilde{T}$, the claim follows from the fact $w(\tilde{T} - y) = w(\tilde{T})$. 
	
	By making use of this claim, the hypothesis on $|S|$ and the facts that $w(\tilde{T}) = w(D^*_{v^*}) > \alpha \cdot w(D^*)$ and $c \cdot \alpha > 3 (\Delta(T) + 1)$, we conclude that $D'$ improves over $D^*$, which is a contradiction.

\medskip \noindent
{\bf Case 2:} $|S| < \frac{2 c}{3}$. We set $y = x$ and construct $D'$ as described at the beginning of the proof.
	
		
		Again, we are trying to reach the contradiction $cost(D') < cost(D^*)$. Recall 
		that $cost(D') \le \sum_{i = 0}^{\delta(y)} cost(D^i) + (\Delta(T) + 1) \cdot w(T)$, so we bound the cost of the trees $D^i$'s.
		
		 By construction we have that $cost(D^0) \le \sum_{u \in T^0} d(r(D^*), l_u) w(u)$. 
		 Now consider some tree $D^i$ for $i \neq 0$. From its construction we have that $\dist{r(D^i)}{l_u}{D^i} \le \dist{r(D^*)}{l_u}{D^*}$ for any node $u \in T^i$. Moreover, for any node $u \in T^i \cap \tilde{T}$ the path $r(D^*) \leadsto l_u$ contains $v^*$ and 
		 therefore it contains at least $c - |S|$ queries to nodes in $T - T_x = T^0$. Then Lemma \ref{subSS} guarantees 
		 that for every $u \in T^i \cap \tilde{T}$ we have $\dist{r(D^i)}{l_u}{D^i} \le \dist{r(D^*)}{l_u}{D^*} - (c - |S|)$.
		
		Weighting these bounds over all nodes in $T$ we have:
		\begin{eqnarray*}
					\sum_{i = 0}^{\delta(y)} cost(D^i) &\le& \sum_{i = 0}^{\delta(y)} \sum_{u \in T^i} d(r(D^*),l_u) w(u) - 
			\sum_{i = 1}^{\delta(y)} \sum_{u \in T^i \cap \tilde{T}} (c - |S|) \cdot w(u)\\ &=&
cost(D^*) - d(r(D^*),l_x) w(x) -   (c - |S|) \cdot (  w(\tilde{T}) - w(x)) \\ &\leq&
		cost(D^*) -   (c - |S|) \cdot w(\tilde{T}),
		\end{eqnarray*}	
where the last inequality is valid because $l_x$ is a descendant of $v^*$ in $D^*$ so that $ d(r(D^*),l_x) \geq c$.
Thus, by combining the upper bound on $cost(D')$ with the previous equation in the  display  we get  that
 $cost(D') \le cost(D^*) - (c - |S|) \cdot w(\tilde{T}) + (\Delta(T) + 1) \cdot w(T)$. 
	By making use of the hypothesis $|S| < \frac{2 c}{3}$ and the facts that 
	$w(\tilde{T}) = w(D^*_{v^*}) > \alpha \cdot w(D^*) = \alpha \cdot w(T)$ and $c \cdot \alpha > 3 (\Delta(T) + 1)$, 
	we conclude that $D'$ improves over $D^*$, which gives the desired contradiction.
		
	\subsection{Proof of Theorem \ref{FPTAS}} \label{app:FPTAS}
The following lemma shows that that the bound on the height of the shortest optimal tree holds even when the weight function is not strictly positive. 
	\begin{lemma} \label{height}
		There is an optimal search tree for $(T,w)$ of height at most $O(\Delta(T) \cdot (\log w(T) + \log n))$.  
	\end{lemma}
	
	\begin{proof}
		Consider an optimal search tree $D^*$ for $(T,w)$. Notice that for any $v \in D^*$, $D^*_{v^*}$ is an optimal search tree for the subtree of $T$ associated with $v$. So we can employ the Lemma \ref{lemma:geoDec} repeatedly and get that for every node $v$ of $D^*$ at a level $l = O(\Delta(T) \cdot \log w(T))$, $w(D^*_v) = 0$. 
			
		Now let $L$ be all the nodes of $D^*$ at level $l$. For each $v \in L$ let $D^v$ be the shortest search tree for the subtree of $T$ associated with node $v$. 
It was proved in \cite{BenFarNew99} that the height of $D^v$ can be upper bounded by $(\Delta(T) + 1) \cdot \log n$. Then we can construct the search tree $D'$ for $T$ as follows: start with $D^*$ and for each $v \in L$ replace $D^*_v$ by $D^v$. Clearly $D'$ has height at most $O(\Delta(T) \cdot (\log w(T) + \log n))$. Moreover, since $w(D^v) = w(D^*_v) = 0$ for all $v \in L$, it follows that $D'$ has the same cost as $D^*$ and hence is optimal.
	\end{proof}

 	\newtheorem{thmFPTAS}[thmFPTASc]{Theorem}
	\begin{thmFPTAS}
		Consider an instance $(T,w)$ to our search problem where $\Delta(T) = O(1)$. Then there is an algorithm for computing an optimal search tree for $(T,w)$ that runs in $poly(n \cdot w(T))$ time. In addition, there is an algorithm for computing an $(1 + \epsilon)$-approximate search tree for $(T,w)$ that runs in $poly(n/\epsilon)$ time.
	\end{thmFPTAS}
	
	\begin{proof}
		The existence of an exact pseudo-polynomial algorithm which runs in $poly(n \cdot w(T))$ time follows from the discussion presented in Section \ref{sec:search} (see {\bf From the DP  algorithm to an FPTAS.}). Thus, we only prove the second claim of the theorem, namely, that our search problem admits an FPTAS.
		
		We claim that the following procedure gives the desired FPTAS:
		
		\begin{enumerate}
			\item Let $W$ be the weight of the heaviest node of $T$, namely $W = \max_{u \in T} \{w(u)\}$. Define $K = \frac{\epsilon \cdot W}{n^2}$ and the weight function $w'$ such that $w'(u) = \lceil w(u) / K \rceil$ for every node $u \in T$. \label{step1FPTAS}
		
			\item Find an optimal search tree $D$ for $(T,w')$ using the pseudo-polynomial algorithm and return $D$. \label{step2FPTAS}
		\end{enumerate}
		
		First we analyze the running time this procedure. Clearly Step \ref{step1FPTAS} takes at most $O(n)$ time. In order to analyze Step \ref{step2FPTAS}, let $W' = \max_{u \in T} \{w'(u)\}$ and notice that $W' = \lceil W / K \rceil \le (n^2) / \epsilon + 1$. Thus, $w'(T) \le n W' \le (n^3) / \epsilon + n$. Then the pseudo-polynomial algorithm employed in Step \ref{step2FPTAS} runs in $poly(n \cdot w'(T)) = poly(n/\epsilon)$. The running time of the whole procedure is then $poly(n/\epsilon)$, as desired.
		
		Now we argue that the solution $D$ returned by the procedure is $(1 + \epsilon)$-approximate for the instance $(T, w)$. Let us make the weights explicit in the cost function, e.g. we denote by $cost(D,w)$ and $cost(D,w')$ the cost of $D$ with respect to the weights $w$ and $w'$. Thus we want to prove that $cost(D, w) \le (1 + \epsilon) cost(D^*, w)$, where $D^*$ is an optimal search tree for $(T, w)$.	
		
		Clearly for each node $u \in T$ we have $K \cdot w'(u) \le w(u) + K$ and hence
	\begin{equation*}
		K \cdot cost(D^*,w') \le cost(D^*, w) + \sum_{u \in T} d(r(D^*), l_u) \cdot K \le cost(D^*, w) + n^2 \cdot K = cost(D^*, w) + \epsilon \cdot W,
	\end{equation*}
	where the last inequality follows from the fact that the distances are trivially upper bounded by $n$. Excluding the trivial case where $T$ is empty, notice that every path in $D^*$ from $r(D^*)$ to a leaf has length at least one. Thus, $cost(D^*,w)$ can be lower bounded by $W$, and the previous displayed inequality gives $K \cdot cost(D^*, w') \le (1 + \epsilon) cost(D^*, w)$. But since $w(u) \le K \cdot w'(u)$ for all $u$, we have that 
	\begin{equation*}
		cost(D,w) \le K \cdot cost(D,w') \le K \cdot cost(D^*, w') \le (1 + \epsilon) cost(D^*, w)\;,
	\end{equation*}
	 where the second inequality follows from the optimality of $D$. Therefore, $D$ is a $(1 + \epsilon)$-approximate search tree for the instance $(T,w)$, which concludes the proof of the theorem.
	\end{proof}

\remove{	
\begin{figure}[!ht]
\begin{center}
\includegraphics[height=10cm]{Realizations5-2}
\caption{The base Realization $D^A$ (left); The reference pseudo-realization $D^A_+$ (center), where 
the place-holder questions are dashed; An (optimal)  Realization w.r.t. 
the exact cover $\{X_1, X_4\}$ (right)---in bold are the questions involved in the configuration changes.
Only the leaves associated to nodes of T with non-zero weights are shown here.} 
\label{fig:realizations}
\end{center}
\end{figure}
}

\section{Polynomiality of the tree search problem for instances of diameter at most 3}

First consider an instance $(T, w)$ of our search problem where $T$ has diameter two, i.e., it is a star. Let us root the star in its center. 
Employing a simple exchange argument it is easy to show that the children of $r(T)$ 
must be queried according to their weights, in decreasing order. 
Thus, an optimal search tree for $(T, w)$ can be built based on any sorting algorithm in $O(n \log n).$
	
Now 
 assume $T$ has diameter 3. 
Notice that the only possible structure for $T$ is the following: there are two nodes $r$ and $r'$ joined by an edge and all other nodes are either adjacent to $r$ or to $r'.$ 
In order to define the questions,  let us take $r$ as the root. 
Let $l$ ($l'$) be the heaviest leaf among the children of $r$ ($r'$). 
It should not be difficult to see that  the root of any optimal search tree must query one of the nodes in the set $\{r',l, l'\}.$ 
This can be proved using a simple exchange argument. 
If $r(D)$ is assigned to $r'$ then its right subtree is  an optimal search tree for $T_{r'}$ and its left subtree is an optimal search tree for $T - T_{r'}.$
If $r(D)$ is assigned to $l$ then its right subtree is a leaf assigned to $l$ and its right subtree is an optimal search tree for $T - l$. 
Analogously, when $r(D)$ is assigned to $l'$ its right subtree is an optimal search tree for $T - l'$. 
Finally, notice that in the first case, both $T_{r'}$ and $T - T_{r'}$ have diameter at most 2.
	
Consider the recursion tree of the above procedure; notice that every subproblem $(T', w)$ has a specific structure: $T'$ is the subtree of $T$ induced by nodes $r$, $r'$, the $i$th heaviest leaf-children of $r$ and the $j$th heaviest children of $r'$ (for some $i,j$). Employing a Dynamic Programming strategy together with an $O(n \log n)$ preprocessing for the two stars centered at $r$ and $r',$   
 it is not difficult to see that each of these $O(n^2)$ problems can be solved in $O(1)$ time. 
This gives an $O(n^2)$ algorithm for finding an optimal search tree for $(T, w)$.

\newpage

\noindent{\Large \bf Figures}

\begin{figure}[!ht]
\begin{center}
\includegraphics[height=4.5cm]{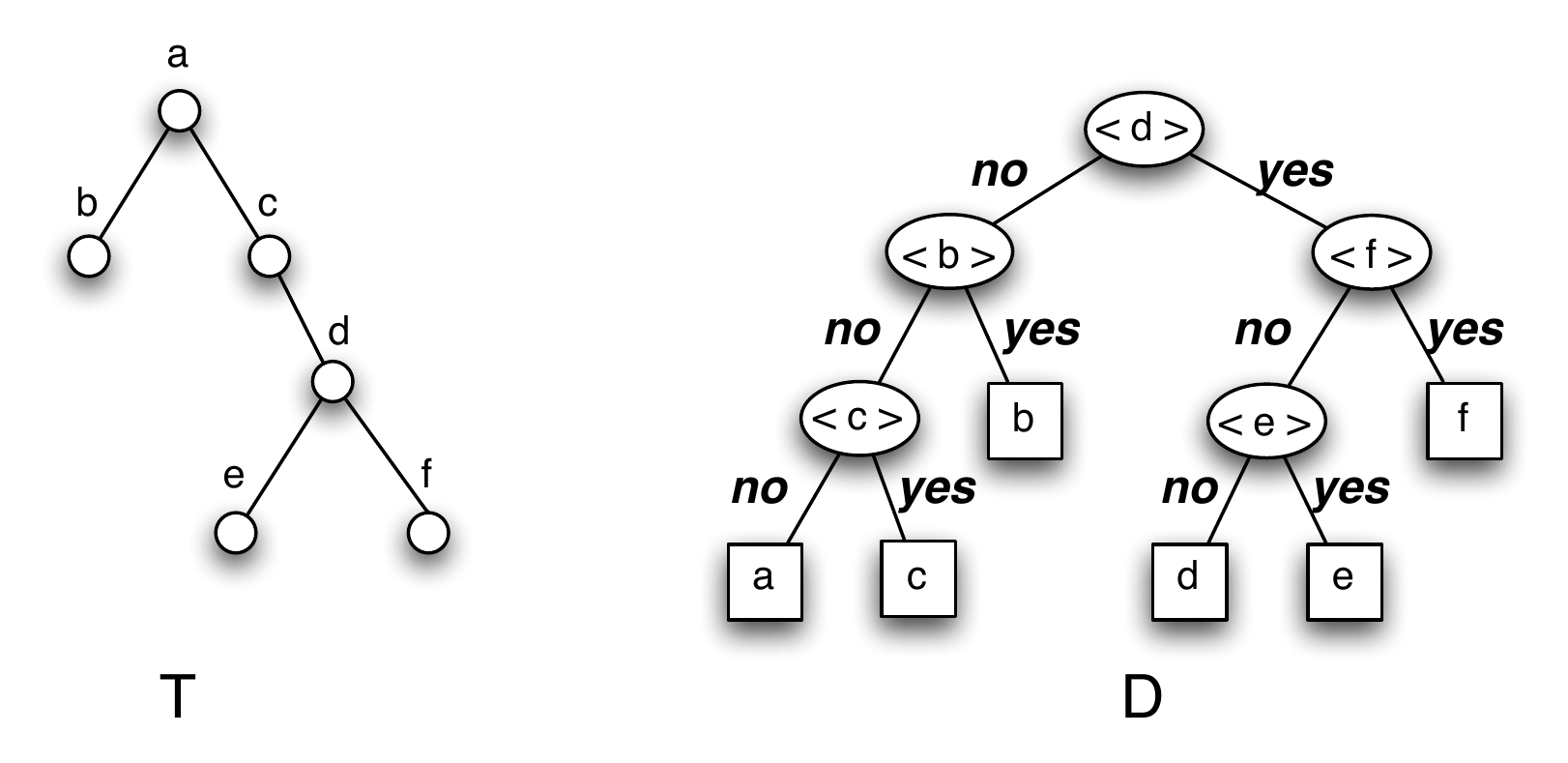}
\caption{(left) The input tree T; (right) a search tree D for  T} 
\label{fig:prob-defi}
\end{center}
\end{figure}

\begin{figure}[!ht]
\begin{center}
\includegraphics[height=6cm]{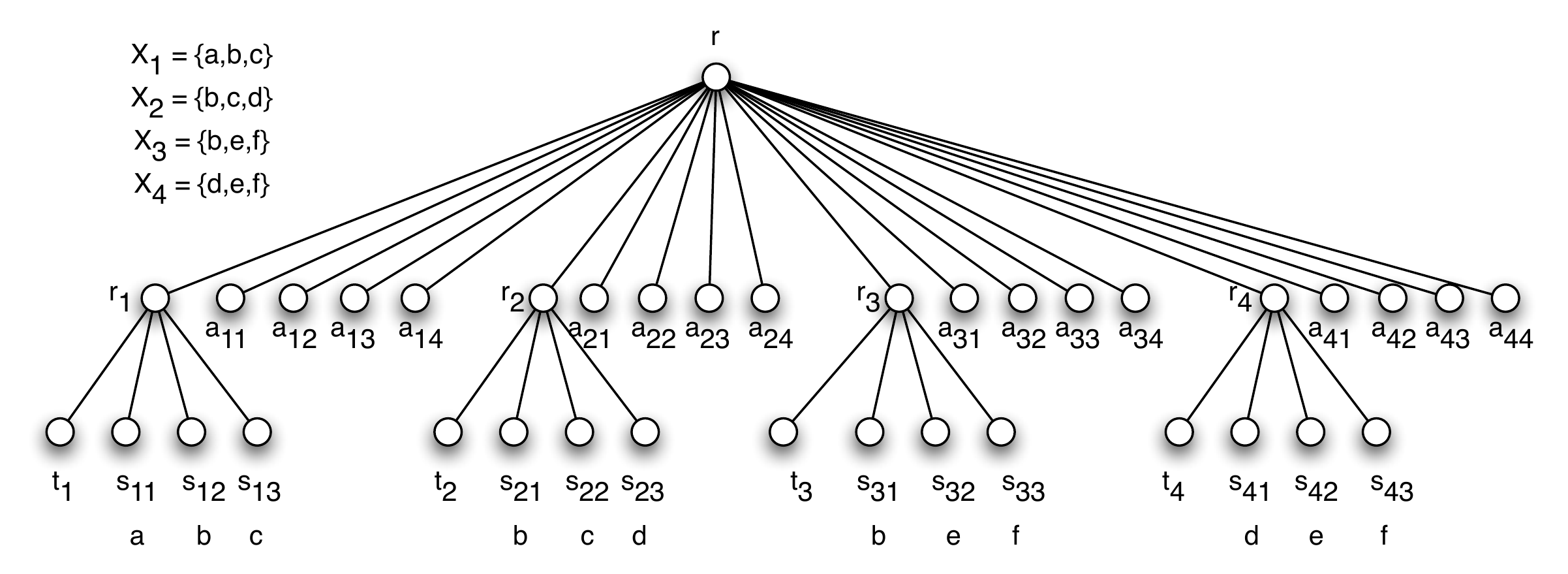}
\caption{The tree obtained from instance ${\mathbb I} = (\{a,b,c,d,e,d,f\}, \{X_1, X_2, X_3, X_4\})$ of 3-bounded
X3C.} 
\label{fig:tree}
\end{center}
\end{figure}

\begin{figure}[!ht]
\begin{center}
\includegraphics[height=6cm]{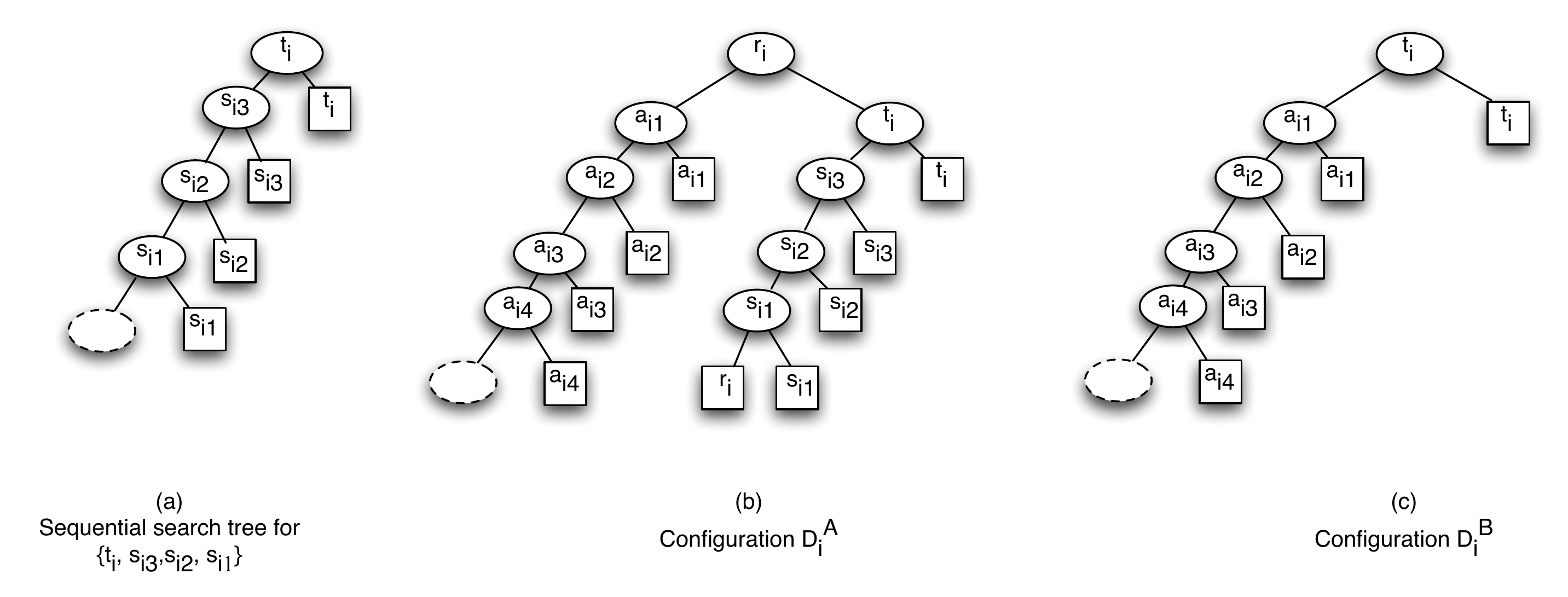}
\caption{The two possible configurations we use for the part of the search tree that 
concerns the  subtree $T_i$ and the leaf $a_{i}$ and a sequential search tree for $T_i.$} 
\label{fig:Config-LinearStrat}
\end{center}
\end{figure}

\begin{figure}[!ht]
\begin{center}
\includegraphics[height=19cm]{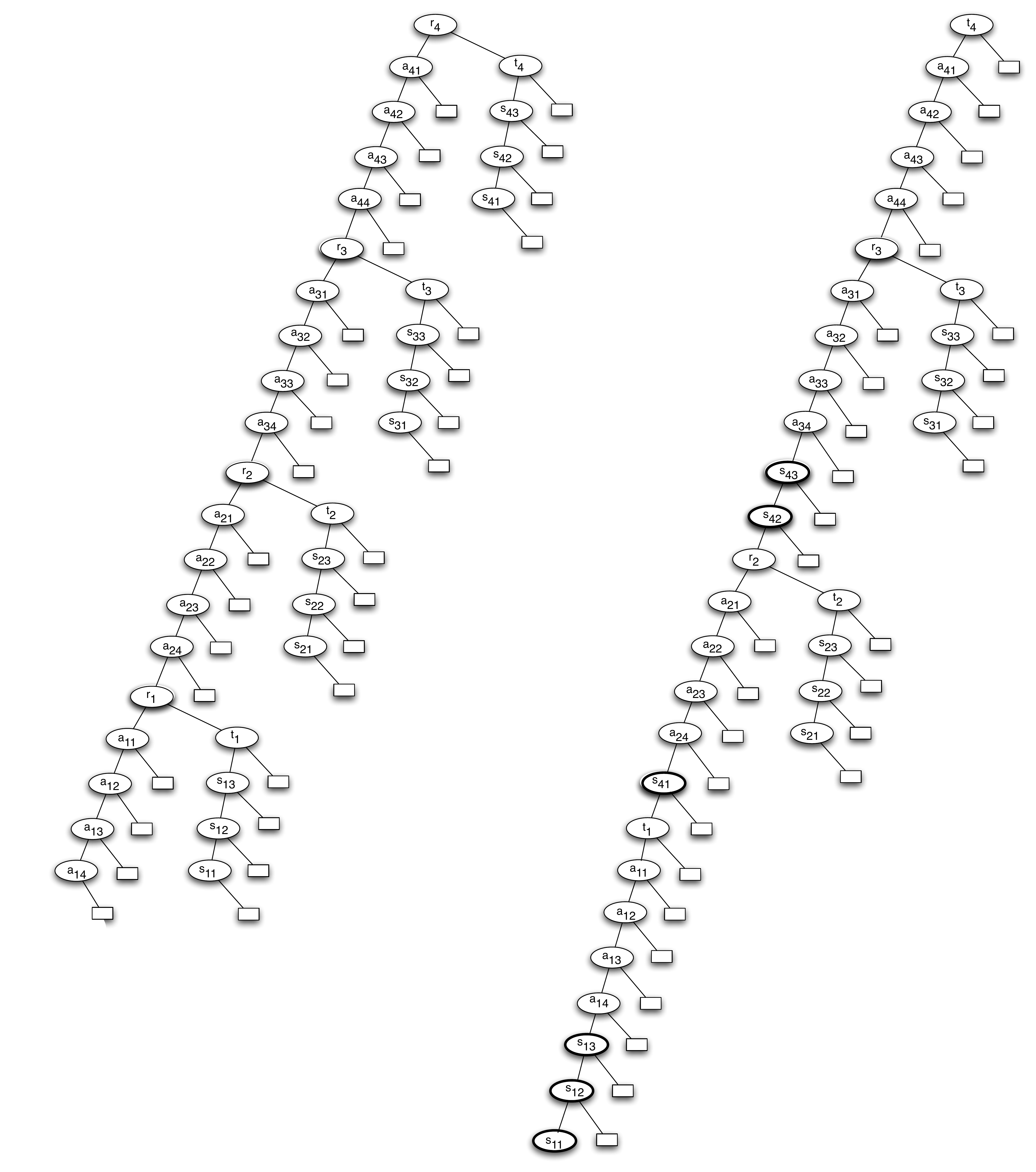}
\caption{Realization $D^A$ (left), and the (optimal)  Realization w.r.t. 
the exact cover $\{X_1, X_4\}$ (right)---in bold are the questions involved in the configuration changes.
Only the leaves associated to nodes of T with non-zero weights are shown here.} 
\label{fig:realizations}
\end{center}
\end{figure}

\noindent
\begin{figure}[!ht]
\includegraphics[height=9.2cm]{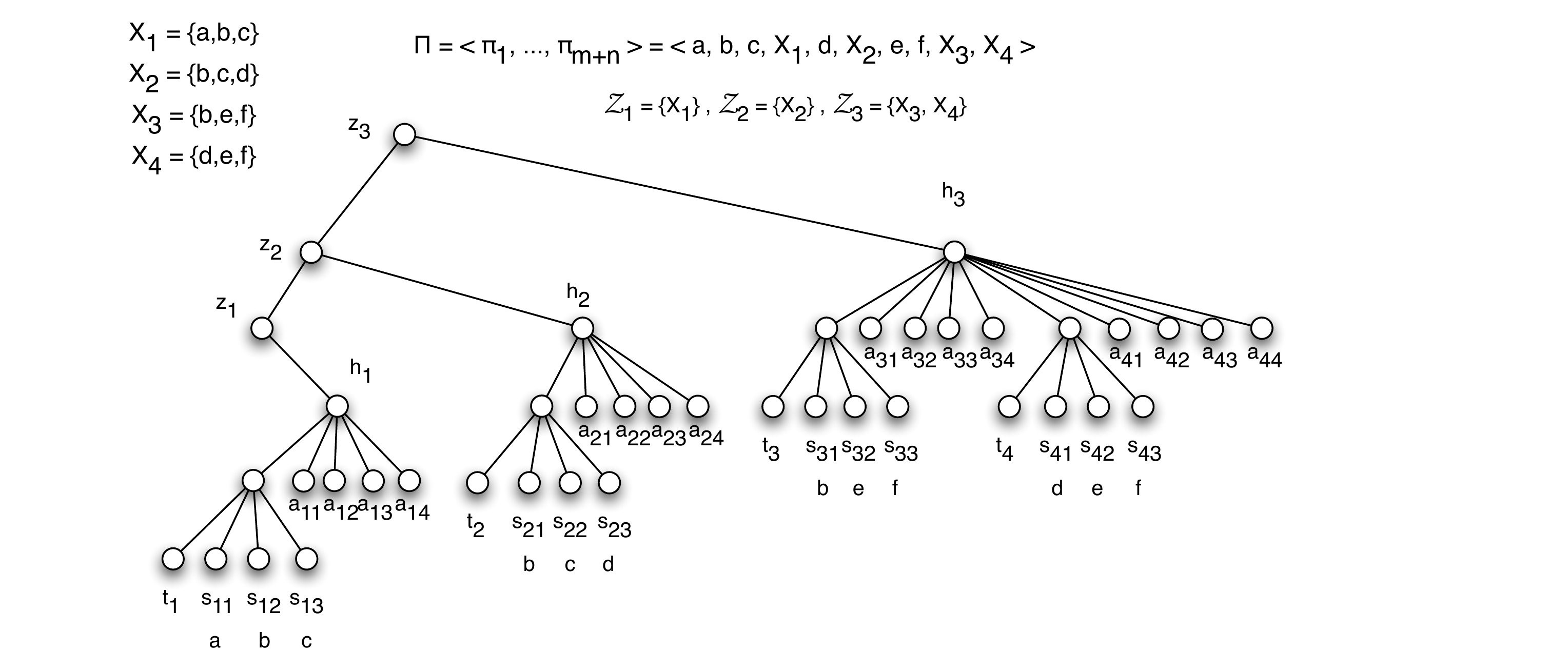}
\caption{The tree $T^b$ obtained from the instance ${\mathbb I} = (\{a,b,c,d,e,d,f\}, \{X_1, X_2, X_3, X_4\})$ of 3-bounded
X3C.} 
\label{fig:tree_b}
\end{figure}

\begin{figure}[!ht]
	\centering
		\includegraphics[scale=0.27]{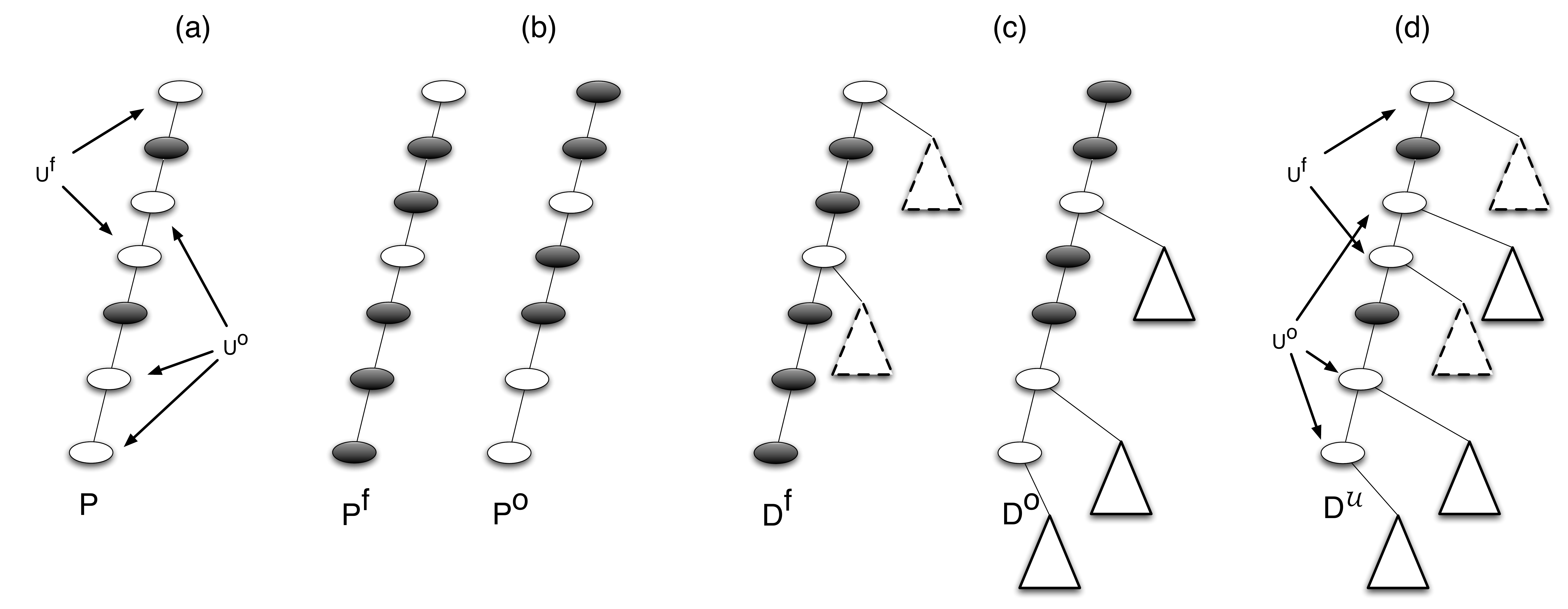}
	\caption{(a) PLP $P$ with partition ${\cal U} = \{U^f, U^o\}$ indicated.
	 The blank nodes are unassigned and the black ones are blocked. 
	 (b) PLP's $P^f$ and $P^o$. (c) The optimal EST's $D^f$ and $D^o$ and (d) the resulting  EST $D^{\cal U}$ constructed by taking the `union' of $D^f$ and $D^o$.}
	\label{fig:case2}
\end{figure}

\begin{figure}[!ht]
	\centering
		\includegraphics[scale=0.39]{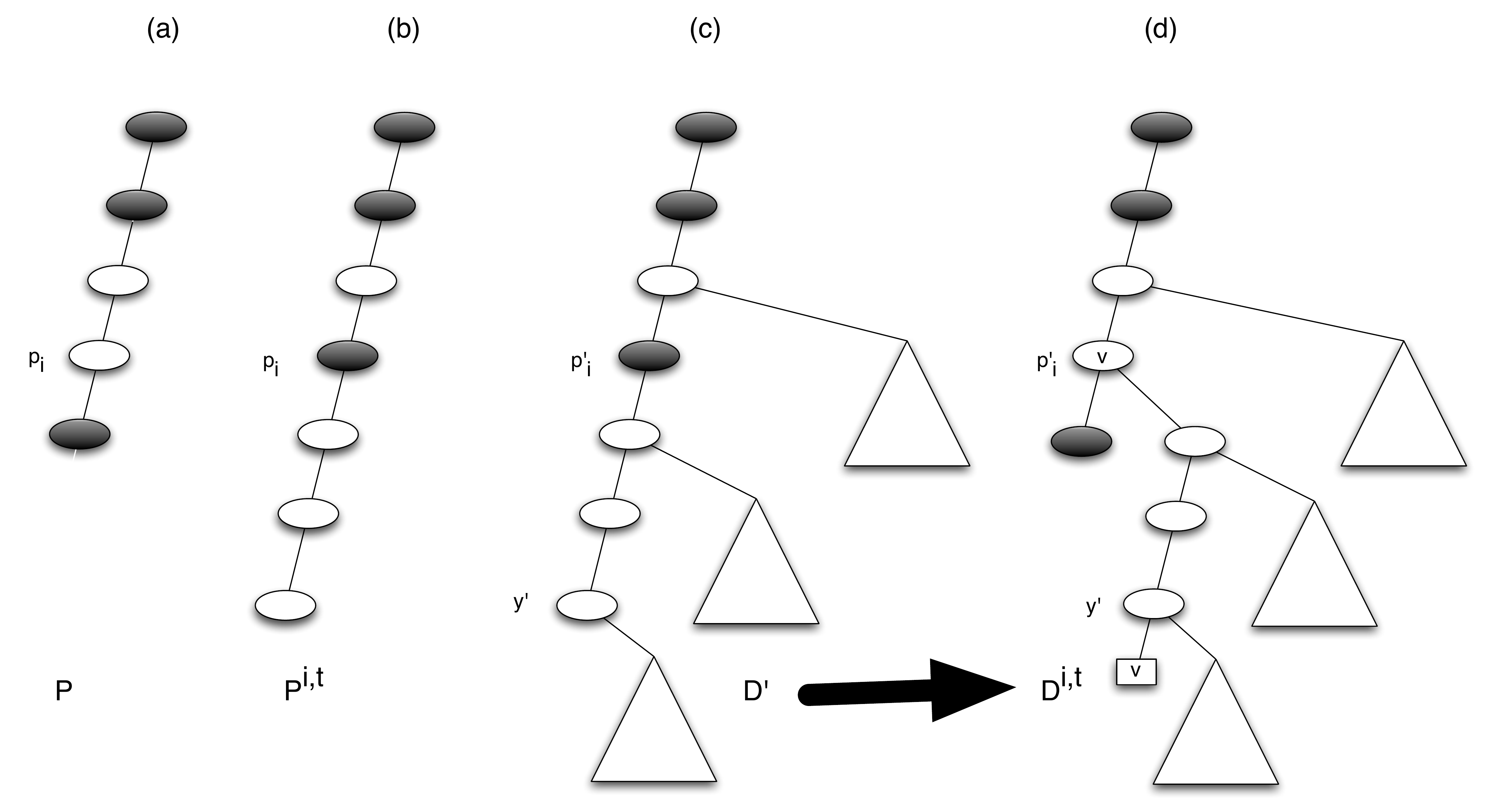}
	\caption{(a) PLP $P$ (b) PLP $P^{i,t}.$  (c)-(d) Construction of $D^{i,t}$---given in picture (d)---starting from an EST $D'$ given in 
	picture (c) ---for ${\cal P}^B(\{T_{c_1(v)}, \dots, T_{c_{\delta(v)}(v)}\}, P^{i,t}).$ 
	}
	\label{fig:case1}
\end{figure}

\end{document}